\documentclass[letterpaper, 10 pt, onecolumn]{ieeeconf}

\usepackage[utf8]{inputenc}
\usepackage[american]{babel}
\usepackage[babel]{csquotes}

\makeatletter \let\NAT@parse\undefined \makeatother 

 \usepackage{natbib}

\usepackage{amscd}
\usepackage{amsmath}
\usepackage{amssymb}

\newtheorem{theorem}{Theorem}
\newtheorem{definition}{Definition}
\newtheorem{example}{Example}
\newtheorem{assumption}{Assumption}
\newtheorem{proposition}{Proposition}
\newtheorem{lemma}{Lemma}
\newtheorem{remark}{Remark}
\newtheorem{algo}{Algorithm}

\usepackage[usenames,dvipsnames]{xcolor}
 \usepackage{array}
\usepackage{xspace}
\usepackage{paralist}
\usepackage{xifthen}
\usepackage{url}
\usepackage{pgf}
\usepackage{tikz}
\usetikzlibrary{trees,decorations,arrows,automata,shadows,positioning,plotmarks,backgrounds,shapes}
\usetikzlibrary{calc,matrix,fit,petri,decorations.pathmorphing,patterns}

\newcommand{\BR}[1]{\left( #1 \right)}
\newcommand\set[1]{\{ #1 \}}
\newcommand\tuple[1]{{\langle #1 \rangle}}

\newcommand{\REFlem}[1]{\text{Lem.~\ref{#1}}}
\newcommand{\REFthm}[1]{\text{Thm.~\ref{#1}}}
\newcommand{\REFdef}[1]{Def.~\ref{#1}}
\newcommand{\REFalg}[1]{Alg.~\ref{#1}}

\newcommand{\REFexp}[1]{Expl.~\ref{#1}}
\newcommand{\REFsec}[1]{Sec.~\ref{#1}}
\newcommand{\REFprop}[1]{Prop.~\ref{#1}}
\newcommand{\REFfig}[1]{Fig.~\ref{#1}}
\newcommand{\REFass}[1]{Ass.~\ref{#1}}
\newcommand{\REFapp}[1]{App.~\ref{#1}}

\newcommand{\eqrefR}[2]{\ensuremath{\langle\eqref{#1}.\text{right}\ifthenelse{\isempty{#2}}{}{.#2}\rangle}}
\newcommand{\eqrefL}[2]{\ensuremath{\langle\eqref{#1}.\text{left}\ifthenelse{\isempty{#2}}{}{.#2}\rangle}}

\let\exampleOrig\endexample
\def\endexample{\hspace*{0pt}\hfill$\triangleleft$\exampleOrig}

\let\definitionOrig\enddefinition
\def\enddefinition{\hspace*{0pt}\hfill$\triangleleft$\definitionOrig}
\newcommand{\ON}[1]{\operatorname{#1}}

\def\clap#1{\hbox to 0pt{\hss#1\hss}}

\newcommand{\DiCases}[4]{\ensuremath{\begin{cases}%
#1,&~#2\\%
#3,&~#4%
\end{cases}}}%
\newcommand{\TriCases}[6]{\ensuremath{\begin{cases}%
#1,&~#2\\%
#3,&~#4\\%
#5,&~#6%
\end{cases}}}%
\makeatletter 
\newif\ifFIRST
\newif\ifSECOND
\let\LISTOP\relax
\newcommand{\List}[4][\;]{#3#1%
	\FIRSTtrue
	\@for\i:=#2\do{%
	\ifFIRST\LISTOP{\i}\FIRSTfalse\else,\LISTOP{\i}\fi%
	}%
	#1#4%
	\let\LISTOP\relax
}
\makeatother
\newcounter{DINGLIST}
\stepcounter{DINGLIST}
\newcommand{\markD}[3][\;\;]{\text{\ding{\the\numexpr171+\theDINGLIST}\stepcounter{DINGLIST}}#1#3}

\newcommand{\SUCHTHAT}{s.t.\xspace}

\makeatletter

\newcommand{\propNeg}{\@ifstar\propNegStar\propNegNoStar}
\newcommand{\propNegStar}[1]{\ensuremath{\left(\propNegNoStar{#1}\right)}}
\newcommand{\propNegNoStar}[2][\cdot]{\ensuremath{\neg\ifthenelse{\isempty{#2}}{#1}{#2}}}

\newcommand{\propConj}{\@ifstar\propConjStar\propConjNoStar}
\newcommand{\propConjStar}[2]{\ensuremath{\left(\propConjNoStar{#1}{#2}\right)}}
\newcommand{\propConjNoStar}[3][\cdot]{\ensuremath{\ifthenelse{\isempty{#2}}{#1}{#2}\wedge\ifthenelse{\isempty{#3}}{#1}{#3}}}

\newcommand{\propDisj}{\@ifstar\propDisjStar\propDisjNoStar}
\newcommand{\propDisjStar}[2]{\ensuremath{\left(\propDisjNoStar{#1}{#2}\right)}}
\newcommand{\propDisjNoStar}[3][\cdot]{\ensuremath{\ifthenelse{\isempty{#2}}{#1}{#2}\vee\ifthenelse{\isempty{#3}}{#1}{#3}}}

\newcommand{\propImp}{\@ifstar\propImpStar\propImpNoStar}
\newcommand{\propImpStar}[2]{\ensuremath{\left(\propImpNoStar{#1}{#2}\right)}}
\newcommand{\propImpNoStar}[3][\cdot]{\ensuremath{\ifthenelse{\isempty{#2}}{#1}{#2}\Rightarrow\ifthenelse{\isempty{#3}}{#1}{#3}}}

\newcommand{\propAequ}{\@ifstar\propAequStar\propAequNoStar}
\newcommand{\propAequStar}[2]{\ensuremath{\left(\propAequNoStar{#1}{#2}\right)}}
\newcommand{\propAequNoStar}[3][\cdot]{\ensuremath{\ifthenelse{\isempty{#2}}{#1}{#2}\Leftrightarrow\ifthenelse{\isempty{#3}}{#1}{#3}}}

\newcommand{\AllQ}{\@ifstar\AllQStar\AllQNoStar}
\newcommand{\AllQStar}[3][\;]{\ensuremath{\left(\forall #2#1.#1#3\right)}}
\newcommand{\AllQNoStar}[3][\;]{\ensuremath{\forall #2#1.#1#3}}
\newcommand{\AllQu}{\@ifstar\AllQuStar\AllQuNoStar}
\newcommand{\AllQuStar}[3][\;]{\ensuremath{\left(\forall^{\infty} #2#1.#1#3\right)}}
\newcommand{\AllQuNoStar}[3][\;]{\ensuremath{\forall^{\infty} #2#1.#1#3}}

\newcommand{\ExQ}{\@ifstar\ExQStar\ExQNoStar}
\newcommand{\ExQStar}[3][\;]{\ensuremath{\left(\exists #2#1.#1#3\right)}}
\newcommand{\ExQNoStar}[3][\;]{\ensuremath{\exists #2#1.#1#3}}

\newcommand{\NExQ}{\@ifstar\NExQStar\NExQNoStar}
\newcommand{\NExQStar}[3][\;]{\ensuremath{\left(\nexists #2#1.#1#3\right)}}
\newcommand{\NExQNoStar}[3][\;]{\ensuremath{\nexists #2#1.#1#3}}

\newcommand{\UniqueQ}{\@ifstar\UniqueQStar\UniqueQNoStar}
\newcommand{\UniqueQStar}[3][\;]{\ensuremath{\left(\exists! #2#1.#1#3\right)}}
\newcommand{\UniqueQNoStar}[3][\;]{\ensuremath{\exists! #2#1.#1#3}}

\newenvironment{propConjA}{\left(\def\unionAtest{1}\begin{array}{@{\if\unionAtest1\gdef\unionAtest{0}\phantom{\wedge}\else\wedge\fi}l@{}}}{\end{array}\right)}

\newenvironment{propDisjA}{\left(\def\unionAtest{1}\begin{array}{@{\if\unionAtest1\gdef\unionAtest{0}\phantom{\vee}\else\vee\fi}l@{}}}{\end{array}\right)}

  \newlength{\SFS@HEIGHT}
  \newlength{\SFS@WIDTH}
  \newcommand{\SplitX}[2]{
	    \settoheight{\SFS@HEIGHT}{$#2$}
	    \settowidth{\SFS@WIDTH}{$#2$}
	    \mbox{\begin{tikzpicture}[baseline=(current bounding box.center)]
	    \node[] (E) at (0,0) {$#1$};
	    \node[inner sep=0pt] (F) at ($(E.south west)+(1ex,-1ex)+(3ex+.5\SFS@WIDTH,-\SFS@HEIGHT)$) {$#2$};
	    \node[] (E) at (0,0) {\phantom{$#1$}};
	    \draw[fill] ($(E.east)+(1ex,0ex)$) circle (.2ex);
	    \draw[-] ($(E.east)+(1ex,0ex)$) -- ($(E.south east)+(1ex,-0.5ex)$) -- ($(E.south west)+(1ex,-0.5ex)$) -- ($(E.south west)+(1ex,-1ex)-(0,\SFS@HEIGHT)$) -- ($(E.south west)+(2.5ex,-1ex)-(0,\SFS@HEIGHT)$);
	    \draw[fill] ($(E.south west)+(2.5ex,-1ex)-(0,\SFS@HEIGHT)$) circle (.2ex);
	    \end{tikzpicture}}}
  \newcommand{\SplitS}[2]{
	    \settoheight{\SFS@HEIGHT}{$#2$}
	    \settowidth{\SFS@WIDTH}{$#2$}
	    \mbox{\begin{tikzpicture}[baseline=(current bounding box.center)]
	    \node[] (E) at (0,0) {$#1$};
	    \node[inner sep=0pt] (F) at ($(E.south west)+(1ex,0.5ex)+(3ex+.5\SFS@WIDTH,-\SFS@HEIGHT)$) {$#2$};
	    \end{tikzpicture}}}	    
  \newcommand{\AllQSplit}[2]{\SplitX{\forall\;#1\;.}{#2}}

\newcommand{\Set}[2][]{\List[#1]{#2}{\{}{\}}}
\newcommand{\VSet}[2][]{\let\LISTOP\val\List[#1]{#2}{\{}{\}}}

\newcommand{\Tuple}[2][]{\List[#1]{#2}{(}{)}}
\newcommand{\VTuple}[2][]{\let\LISTOP\val\List[#1]{#2}{(}{)}}

\newcommand{\UNION}{\@ifstar\UNIONStar\UNIONNoStar}
\newcommand{\UNIONStar}[2]{\ensuremath{\left(\UNIONNoStar{#1}{#2}\right)}}
\newcommand{\UNIONNoStar}[2]{\ensuremath{\ifthenelse{\isempty{#1}}{\cdot}{#1}\cup\ifthenelse{\isempty{#2}}{\cdot}{#2}}}

\newcommand{\UNIOND}{\@ifstar\UNIONDStar\UNIONDNoStar}
\newcommand{\UNIONDStar}[2]{\ensuremath{\left(\UNIONDNoStar{#1}{#2}\right)}}
\newcommand{\UNIONDNoStar}[2]{\ensuremath{\ifthenelse{\isempty{#1}}{\cdot}{#1}\uplus\ifthenelse{\isempty{#2}}{\cdot}{#2}}}

\newcommand{\SETMINUS}{\@ifstar\SETMINUSStar\SETMINUSNoStar}
\newcommand{\SETMINUSStar}[2]{\ensuremath{\left(\SETMINUSNoStar{#1}{#2}\right)}}
\newcommand{\SETMINUSNoStar}[2]{\ensuremath{\ifthenelse{\isempty{#1}}{\cdot}{#1}\setminus\ifthenelse{\isempty{#2}}{\cdot}{#2}}}

\newcommand{\INTERSECT}{\@ifstar\INTERSECTStar\INTERSECTNoStar}
\newcommand{\INTERSECTStar}[2]{\ensuremath{\left(\INTERSECTNoStar{#1}{#2}\right)}}
\newcommand{\INTERSECTNoStar}[2]{\ensuremath{\ifthenelse{\isempty{#1}}{\cdot}{#1}\cap\ifthenelse{\isempty{#2}}{\cdot}{#2}}}

\newcommand{\CARTPROD}{\@ifstar\CARTPRODStar\CARTPRODNoStar}
\newcommand{\CARTPRODStar}[2]{\ensuremath{\left(\CARTPRODNoStar{#1}{#2}\right)}}
\newcommand{\CARTPRODNoStar}[2]{\ensuremath{\ifthenelse{\isempty{#1}}{\cdot}{#1}\times\ifthenelse{\isempty{#2}}{\cdot}{#2}}}

\newcommand{\FINCOUNT}{\@ifstar\FinCountStar\FinCountNoStar}
\newcommand{\FinCountStar}[1]{\ensuremath{\#(\ifthenelse{\isempty{#1}}{\cdot}{#1})}}
\newcommand{\FinCountNoStar}[1]{\ensuremath{\#\left(\ifthenelse{\isempty{#1}}{\cdot}{#1}\right)}}

\makeatother 

\newcommand{\sconc}{\cdot}

\newcommand{\inps}{\hspace{-0.1cm}\in\hspace{-0.1cm}}
\newcommand{\eqps}{\hspace{-0.1cm}=\hspace{-0.1cm}}
\newcommand{\plps}{\hspace{-0.05cm}+\hspace{-0.05cm}}
\newcommand{\mips}{\hspace{-0.05cm}-\hspace{-0.05cm}}

\newcommand{\parfun}{\ensuremath{\ON{\rightharpoonup}}}
\newcommand{\fun}{\ensuremath{\ON{\rightarrow}}}
\newcommand{\SetComp}[3][]{\{#1#2#1\mid#1#3#1\}}
\newcommand{\SetCompX}[3][]{\left\{#1#2#1\middle\vert#1#3#1\right\}}

\newcommand{\Nb}{\ensuremath{\mathbb{N}}} 

\newcommand{\twoup}[1]{\ensuremath{2^{#1}}}

\newcommand{\dom}[1]{\ensuremath{\mathrm{dom}(#1)}}
\newcommand{\domp}[1]{\ensuremath{\mathrm{dom}^+(#1)}}

\renewcommand{\ll}[1]{\ensuremath{|_{[#1]}}}

\newcommand{\maxk}[1]{\ensuremath{\mathsf{end}({#1})}}
\newcommand{\maxw}[1]{\ensuremath{\lceil#1\rceil}}

\newcommand{\Y}[1]{\ensuremath{Y^{#1}}}
\newcommand{\Yli}[2]{\ensuremath{Y^{#1}_{#2\rceil}}}
\newcommand{\Yla}[2]{\ensuremath{Y^{#1}_{#2\lfloor}}}

\newcommand{\X}[1]{\ensuremath{X^{#1}}}

\newcommand{\YI}[2]{\ensuremath{\ifthenelse{\isempty{#2}}{Y^{#1}_I}{Y^{#1}_{I,#2}}}}
\newcommand{\XI}[2]{\ensuremath{\ifthenelse{\isempty{#2}}{X^{#1}_I}{X^{#1}_{I,#2}}}}
\newcommand{\I}[1]{\ensuremath{\mathcal{I}^{#1}}}

\newcommand{\rx}[2]{\ensuremath{\mathfrak{r}^{#1}_{#2}}}

\newcommand{\G}{\ensuremath{G}}

\newcommand{\Glall}{\ensuremath{\mathbb{G}}}

\newcommand{\TrS}[2]{\ensuremath{\ifthenelse{\isempty{#2}}{\rho^{#1}}{\rho_{#2}^{#1}}}}
\newcommand{\TrE}[2]{\ensuremath{\ifthenelse{\isempty{#2}}{\delta^{#1}}{\delta_{#2}^{#1}}}}
\newcommand{\TrSt}[2]{\ensuremath{\ifthenelse{\isempty{#2}}{\tilde{\rho}^{#1}}{\tilde{\rho}_{#2}^{#1}}}}
\newcommand{\TrEt}[2]{\ensuremath{\ifthenelse{\isempty{#2}}{\tilde{\delta}^{#1}}{\tilde{\delta}_{#2}^{#1}}}}
\newcommand{\PhiallR}{\ensuremath{[\varphi]}}
\newcommand{\Alphaall}{\ensuremath{[\eass]}}
\newcommand{\rhoall}{\ensuremath{[\breve{p}]_{\play{}}}}
\newcommand{\phisconc}[3]{\ensuremath{\phi^{#1}_{#2}({#3})}}

\newcommand{\eass}{\ensuremath{\zeta}}

\newcommand{\Compl}{\ensuremath{\mathsf{CompliantPlays}}}

\newcommand{\WinPlays}{\ensuremath{\mathsf{WinningPlays}}}
\newcommand{\WinStrat}{\ensuremath{\mathsf{WinningStrategies}}}
\newcommand{\AdmStrat}{\ensuremath{\mathsf{AdmissibleStrategies}}}
\newcommand{\StrategyE}{\ensuremath{\mathcal{S}^e}}
\newcommand{\StrategyS}{\ensuremath{\mathcal{S}^s}}

\newcommand{\play}[1]{\ensuremath{\pi^{#1}}}
\newcommand{\plx}[1]{\ensuremath{x^{#1}}}
\newcommand{\ply}[1]{\ensuremath{y^{#1}}}
\newcommand{\pld}[1]{\ensuremath{\pi^{#1}_{\downarrow}}}
\newcommand{\plxd}[1]{\ensuremath{x^{#1}_{\downarrow}}}

\newcommand{\Play}[1]{\ensuremath{\mathcal{\G}^{#1}}}

\newcommand{\plP}[1]{\ensuremath{\breve{\pi}^{#1}}}

\newcommand{\plPd}[1]{\ensuremath{\breve{\pi}^{#1}_{\downarrow}}}
\newcommand{\plPdm}[1]{\ensuremath{\breve{p}^{#1}_{\downarrow}}}

\newcommand{\plPx}[1]{\ensuremath{\breve{x}^{#1}}}
\newcommand{\plPxd}[1]{\ensuremath{\breve{x}^{#1}_{\downarrow}}}

\newcommand{\plPy}[1]{\ensuremath{\breve{y}^{#1}}}
\newcommand{\PlP}{\ensuremath{\breve{\Pi}}}

\newcommand{\plg}[1]{\ensuremath{\breve{\gamma}^{#1}}}

\newcommand{\f}[1]{\ensuremath{f^{#1}}}
\newcommand{\g}[1]{\ensuremath{g^{#1}}}

\newcommand{\Sol}[2]{\ensuremath{\ON{Sol}^{#1}\ifthenelse{\isempty{#2}}{}{\BR{#2}}}}

\newcommand{\Rx}[1]{\ensuremath{\alpha_{e}^{#1}}}
\newcommand{\Ry}[1]{\ensuremath{\alpha_{s}^{#1}}}
\newcommand{\Rxup}[1]{\ensuremath{\alpha_{e}^{#1^{\uparrow}}}}

\newcommand{\Ryup}[1]{\ensuremath{\alpha_{s}^{#1^{\uparrow}}}}

\newcommand{\lmax}{\ensuremath{L}}

\newcommand{\GotStuck}[1]{\ensuremath{\textcolor{BrickRed}{\ON{GotStuck}^{#1}}}}

\newcommand{\UnReal}[1]{\ensuremath{\textcolor{BrickRed}{\ON{UnRealizable}^{#1}}}}
\newcommand{\Win}[1]{\ensuremath{\textcolor{BrickRed}{\ON{Win}^{#1}}}}
\newcommand{\Done}[1]{\ensuremath{\textcolor{BrickRed}{\ON{Done}^{#1}}}}

\IEEEoverridecommandlockouts

\title{\LARGE \bf Dynamic Hierarchical Reactive Controller Synthesis}

\author{Anne-Kathrin Schmuck, Rupak Majumdar
\thanks{A.-K. Schmuck and Rupak Majumdar are with the Max Planck Institute for Software Systems (MPI-SWS), Kaiserslautern, Germany. {\tt\small \{akschmuck,rupak\}@mpi-sws.org}}
}

\begin{document}

%

\maketitle

\begin{abstract}
In the formal approach to reactive controller synthesis, a symbolic controller for a possibly hybrid system 
is obtained by algorithmically computing a winning strategy in a two-player game. 
Such game-solving algorithms scale poorly as the size of the game graph increases. 
However, in many applications,
the game graph has a natural hierarchical structure. 
In this paper, we propose a modeling formalism and a synthesis algorithm that exploits this hierarchical structure for more scalable synthesis. 

We define local games on hierarchical graphs 
as a modeling formalism which decomposes a large-scale reactive synthesis problem in two dimensions.
First, the construction of a hierarchical game graph introduces abstraction layers, where 
each layer is again a two-player game graph.
Second, every such layer is decomposed into multiple local game graphs, 
each corresponding to a node in the higher level game graph. 
While local games have the potential to reduce the state space for controller synthesis,
they lead to more complex synthesis problems where strategies computed for one local game can impose
additional requirements on lower-level local games. 

Our second contribution is a procedure to construct a dynamic controller for local game graphs over hierarchies.
The controller computes assume-admissible winning strategies that satisfy local specifications
in the presence of environment assumptions, and dynamically updates specifications and strategies due
to interactions between games at different abstraction layers at each step of the play.
We show that our synthesis procedure is sound: the controller constructs a play which satisfies all local specifications.
We illustrate our results through an example controlling an autonomous robot in a 
known, multistory building.
\end{abstract}

\section{Introduction}

Algorithmic reactive synthesis has recently emerged as a robust methodology to design
correct-by-construction controller for specifications given in temporal logics 
\citep[see, e.g.,][]{GirardPappas2009,TabuadaBook,KloetzerBelta_2008,WolffTopcuMurray_2013b,WongFinucaneKressGazit_2013}.
In this technique, one solves a two-player discrete-time game on a graph between the \emph{system}
and the \emph{environment} players, where the winning condition is specified in linear-time temporal logic.
The game graph is usually obtained as a discrete abstraction of the underlying, possibly continuous or hybrid, dynamics. 
A winning strategy for the system player in such a game can be computed by algorithmic techniques from reactive synthesis \citep{Zielonka,EJ91}.
Such a system winning strategy gives a discrete controller, which can usually be refined
to a continuous controller using primitives from continuous control. 
This controller synthesis methodology has been implemented in symbolic tools \citep{TuLiP,Pessoa,LTLMoP} and was
successfully applied in a number of case studies, e.g., by \citet{WongFinucaneKressGazit_2013,WongpiromsarnTopcuMurray_2010}.

The two major concerns in the application of reactive synthesis to large problems is 
\begin{inparaenum}[(i)]
 \item the poor scalability of
the symbolic game solving algorithms with increasing size of the game graph, and
\item the limited existence of winning strategies against adversarial environment players in realistic settings.
\end{inparaenum}
In this paper, we address these challenges by extending the scope of reactive synthesis for control by
\begin{inparaenum}[(i)]
 \item introducing \emph{local game graphs over hierarchies} as a new decomposed model,
 \item formalizing \emph{hierarchical reactive games} over such models, and 
 \item proposing a sound \emph{reactive controller synthesis algorithm} for such games. 
\end{inparaenum}
This algorithm allows for \emph{dynamic specification changes} and uses the construction of \emph{assume-admissible winning strategies}  
\cite{BrenguierRaskinSankur_ArXiv_2015} to explicitly model and use environment assumptions.

\paragraph{Local Game Graphs over Hierarchies}\label{sec:intro:LGG}
The modeling formalism introduced in this paper allows to exploit the intrinsic
\emph{hierarchy} and \emph{locality} of a given large-scale system.
This decomposes the controller synthesis problem into multiple small ones.
Here, hierarchy means that the game graph allows for the introduction of abstract layers. 
Locality means that a state at a higher layer naturally corresponds to a sub-arena of the game graph at the next lower 
layer which is independent from all the other games at the same layer.

As an example, consider an autonomous robot traversing the floors of a building.
The lowest layer of the game graph, the game under consideration in existing reactive synthesis techniques,
would consist of states defined by grids giving the location and velocity of the robot in each room and each floor
of the building, together with additional predicates, such as
the location of obstacles, whether the robot is carrying something, or the open-closed status of each door.
However, there is a natural hierarchy of abstractions: at the highest layer, we care only about the floors and may
ask the robot to move from one floor to another;
in the next layer, we would like to know the specific room it is in and specify which room to go next,
and only within the context of a room, we may care about 
where exactly the robot is and where it has to go next.  
To model this hierarchy, we introduce a set of layers on top of a game graph, each being a game graph itself,
where a state at a higher layer (e.g. a room) corresponds to a sub-arena of the game graph at the next 
lower layer (i.e., all states located inside this room), modeling locality within the hierarchy. 
%

Such hierarchical and local decompositions are also heuristically applied in robotics. 
Examples are general modeling frameworks, such as hierarchical task-networks (HTN) \citep{HTN_1995} 
or Object-Action Complexes (OAC) \citep{OAC_2009}, or particular
software architectures for incorporating long term tasks and short time motion planning for robots 
\citep{KaelblingLozanoPerez_2011,SrivastavaFangRianoChitnisRusselAbbeel_2014,StockMansouriPecoraHertzberg_2015}. 
One could view our abstraction layers, their interaction, and the system dynamics as an equivalent formalism to model
task networks. 
Our controller synthesis algorithms should also apply to design controllers in these formalisms.
To the best of our knowledge, the problem of correct-by-construction synthesis for temporal logic specifications (beyond reachability)
in the presence of environment assumptions has not been considered by these other formalisms.

Hierarchical approaches for control exist for other correct-by-construction controller synthesis techniques in the control community, 
such as supervisory control \citep[e.g.,][]{schmidt2008nonblocking}, hybrid control \citep[e.g.,][]{RaischMoor2005}, 
or continuous control \citep[e.g.,][]{PappasLafferriereSasty_2000},  but these can usually not handle temporal logic specifications.

In many large-scale projects using reactive controller synthesis, such as autonomous vehicles \citep{HessAlthoffSattel_2014,WongpiromsarnTopcuMurray_2012}
and autonomous flight control \citep{KooSastry_2002}, similar hierarchical and local decompositions are implicitly and informally
performed. However, there is no clear theoretical model
connecting ``low-layer'' reactive control and ``higher layer'' task planning in their work, which is provided by our approach.

\paragraph{Hierarchical Reactive Games}\label{sec:intro:HRG}

To effectively use the constructed hierarchies of local game graphs for reactive controller synthesis, 
we assume that the specification is also decomposed into a set of local requirements, each restricted to one sub-arena of a particular layer, 
together with one \enquote{global} game at the highest layer. 
While such a decomposition is not guaranteed to exist for a given specification, it 
is usually quite natural to exist for specifications over large scale systems with intrinsic hierarchy and locality. 
For example, for the robot, one may consider the specifications:
\begin{inparaenum}[(i)]
 \item a floor-layer task \enquote{visit all floors},
 \item a room-layer task \enquote{visit all rooms} for each floor, and 
 \item a low layer task \enquote{if there is an empty bottle [in the current room], reach it and pick it up} for every room.
\end{inparaenum}

Synthesizing winning strategies for local games over hierarchies w.r.t.\ such sets of local specifications becomes 
challenging due to the interplay between layers both in a bottom-up and a top-down manner. 
The top-down interplay results because applying a strategy in a higher layer introduces additional specifications for 
the lower layer. 
For example, a requested move from one room to an adjacent one requires the local game in this room 
to fulfill a reachability specification in addition to its local specification.
The bottom-up interplay results from the fact that moves in the lowest layer game correspond to moves in all higher layers which might change the strategy.
For example, consider a room with two doors to two different adjacent rooms. 
The higher layer strategy may initially pick one door to continue. 
However, if this door gets closed before it was reached in the lower layer game, the higher layer strategy might ask to reach the second door instead.
Thus, in each local game, winning objectives are generated \emph{dynamically}, based on the strategy at a higher layer, the 
local specification for the local game and the current system and environment state in the lowest layer.

Intuitively, such interactive hierarchical games are similar to pushdown and modular games \citep{Walukiewicz96,AlurLaTorreMadhusudan_2003b,DeCrescenzoLaTorre_2013}, where the local state and the stack determine which (single) local game is played at a particular time point. In contrast, we always play one local game in every layer simultaneously, where visited states in different layers are projections of one another. Therefore, a move in one layer has to be correlated with the games at all other layers at all time steps, giving the dynamic interaction described above. 

Our work also relates naturally to abstraction and refinement techniques in game solving, \citep[e.g.,][]{CC77,HenzingerMajumdarMangRaskin_2000,AbadiLamport_1991}, which map \enquote{concrete} game structures with \enquote{abstract} ones with more abstract timing, to solve a single game for a global specification using different abstraction layers.
In comparison, we propose a hierarchical structure where every system state is refined to a whole new local sub-game, having its own specification. Therefore, the game in the higher layer does only proceed for one step once the lower layer local sub-game is completed. In this sense we are "stitching" together solutions of local games in the lowest layer in a particular way which is determined by higher level games, to obtain a solution to the global game.

%

\paragraph{Dynamical Controller Synthesis}
Given the hierarchical reactive games described above, we propose a reactive controller synthesis algorithm 
to solve such games, which allows for \emph{dynamic specification changes} at each step of the play. 
Intuitively, the controller solves the dynamically constructed local games online and \enquote{stitches} 
their solutions together following the rules of the hierarchical game.
Notice that a strategy computed at a level imposes additional conditions on games at lower levels; thus,
we use a dynamic controller synthesis algorithm that updates the strategies as the game progresses.

In principle, any algorithm which calculates a winning strategy for a two-player game can be used as a building 
block to solve local games \citep[e.g.,][]{Zielonka,EJ91,KV01,Finkbeiner,KW12}.
However, these algorithms calculate winning strategies against \emph{any} environment behavior. 
In most applications, such as our robot example, the requirement that the system wins against any environment strategy is 
too strong. 
For instance, in the robot example it is possible, but very unlikely, that an employee keeps an office door closed forever to prevent 
the robot to fulfill its task. 
Therefore, recently, assumptions on the environment behavior, which model \enquote{likely}
behaviors of the latter, were considered to constrain the synthesis problem 
(see \cite{BloemEhlersJacobsKoenighofer_2014} and \cite{BrenguierRaskinSankur_ArXiv_2015} for a detailed overview of recent results). 
Intuitively, the constrained synthesis problem then asks if the system can win provided that the environment only 
behaves according to its assumptions. 
One type of strategies solving this problem are assume-admissible winning strategies by \cite{BrenguierRaskinSankur_ArXiv_2015}. 
As this is the most expressive available technique to deal with environment assumptions known by the 
authors, we use their synthesis algorithm as a building block in our algorithm.

We prove that, whenever the environment meets its assumptions and all dynamically generated local games have a 
solution, our dynamical synthesis algorithm generates a winning hierarchical play for a given specification, 
i.e., the algorithm is sound. 
If these assumptions do not hold, we show that the play gets stuck but does not violate the specification up to this point. 

The dynamic nature of our controller is also similar to the receding horizon strategies proposed 
by \citet{WongpiromsarnTopcuMurray_2012,VasileBelta_2014}, which translate long term goals into current 
local reachability specifications. 
This approach allows for a particular two-layer hierarchy and uses time horizons to decompose the synthesis problem locally. 
However, the general intrinsic hierarchical and local decomposability of a synthesis problem and the 
interaction of multiple abstract games is not formally exploited.
In our presentation, our control synthesis algorithm solves local games completely; 
however, we can also use a receding horizon controller for each local game.

This paper was motivated by a systems project to build 
an end-to-end autonomous robotic telepresence system. For the scale of this model, existing reactive synthesis techniques would not work.
However, the overall problem has a natural decomposition captured by our proposed model.
While this paper focuses on the theoretical foundations of such a formal model and its reactive controller synthesis, we will discuss the implementation and systems aspects of our technique in a different paper.

\section{Preliminaries}\label{sec:Prelim}

In this section we first introduce notation and recall existing results from reactive synthesis. 
Then we discuss a detailed example to motivate our work.
%

%
%

\subsection{Reactive Synthesis Revisited}



%

\paragraph{Notation}
For a set $W$, we denote by $W^*$, $W^+$, and $W^\omega$ the set of finite sequences, non-empty finite sequences,
and infinite sequences, respectively, over $W$.
We write  $W^\infty = W^* \cup W^\omega$. 
For $w\in W^*$, we write $|w|$ for the length of $w$; the length of $w\in W^\omega$ is $\infty$. 
We define $\dom{w} = \Set{0,\ldots, |w|-1}$ if $w\in W^*$, and $\dom{w} = \Nb$ if $w\in W^\omega$. We denote by $\domp{w}=\SETMINUS{\dom{w}}{\Set{0}}$ the positive domain of $w$.
For $k\in \dom{w}$ we write $w(k)$ for the $k$th symbol of $w$, $\maxw{w}=w(|w|-1)$ for the last symbol of $w$, 
and $w|_{[0,k]}$ for the restriction of $w$ to the domain $[0,k]$. Furthermore, $w\sconc w'$ for $w\in W^*$ and $w'\in W^{\infty}$ denotes the concatenation of two strings. 
The \textit{prefix relation} on strings is defined by $w\sqsubseteq w'$ if
${\ExQ{w''\in W^*}{w\sconc w''=w'}}$. 
Given a set of strings $\varphi\subseteq W^\infty$, we denote by $\overline{\varphi}=\varphi\cup\SetComp{w\in W^*}{\ExQ{w'\in \varphi}{w\sqsubseteq w'}}$ the set of strings in $\varphi$ and all their \emph{finite} prefixes. Slightly abusing notation, we denote by $\overline{w}$ the set $\overline{\Set{w}}$ of all prefixes of the string $w\in W^\infty$.
\paragraph{Two-Player Games}
 
A two-player \emph{game graph} 
$\G=\Tuple{\X{},\Y{},\TrE{}{},\TrS{}{}}$
between environment and system consists of  
a set of environment states $\X{}$,
a set of system states $\Y{}$, 
an environment transition map $\TrE{}{}:~\X{}\times\Y{}\rightarrow \twoup{\X{}}$, and 
a system transition map $\TrS{}{}:~\X{}\times\Y{}\rightarrow \twoup{\Y{}}$.
We assume $\G$ is serial, i.e., $\TrE{}{}$ and $\TrS{}{}$ map each input
to non-empty sets. 
%
A sequence $\play{}\in \BR{\X{}\times\Y{}}^\infty$ 
with $\play{}(k)=\Tuple{x(k),y(k)}$ for all $k\in\dom{\play{}}$ is 
called a \emph{play} in $\G$ 
if 
 \begin{align}
 &\AllQ{k\in \domp{\pi}}{
 \begin{propConjA}
  x(k)\in\TrE{}{}\BR{x(k-1),y(k-1)}\\
  y(k)\in\TrS{}{}\BR{x(k),y(k-1)}
  \end{propConjA}}.\label{equ:playp_def:b}
  \end{align}
A play $\play{}$ is \emph{finite} if $|\play{}|<\infty$ and \emph{infinite} otherwise.
The set of all plays is denoted by $\Play{}$.

We model a \emph{winning condition} in a two-player game as a set of plays $\varphi \subseteq \Play{}$. 
This set can be represented in different ways, e.g., by an LTL formula or by an $\omega$-automaton. 
While our results do not assume a particular representation, 
the latter will determine the algorithm needed to 
solve the two-player game.
 
Given a game graph $\G$, a set of initial strings $\I{}=\Tuple{\X{}\times\Y{}}^+\subseteq\Play{}$
and a winning condition $\varphi \subseteq \Play{}$, 
the tuple $\Tuple{\G,\I{},\varphi}$ is called a \emph{game} 
on $\G$ w.r.t. $\I{}$ and $\varphi$. 
A play $\play{}\in\Play{}$ is \emph{winning} (resp. \emph{possibly winning}) 
for $\Tuple{\G,\I{},\varphi}$
if there exists an $n\in\dom{\play{}}$ s.t. $\play{}|_{[0,n]} \in \I{}$ and 
 $\play{}\in \varphi$ (resp. $\play{}\in \overline{\varphi}$). 
 We denote the set of all winning and possibly winning plays for $\Tuple{\G,\I{},\varphi}$ by $\WinPlays\Tuple{\G,\I{},\varphi}$ and $\WinPlays\Tuple{\G,\I{},\overline{\varphi}}$, respectively.

\paragraph{Strategies}
A \emph{system strategy} is a \emph{partial} function $\f{}: (\X{}\times\Y{})^+ \times \X{} \parfun \Y{}$ such that\footnote{Here, we write $\maxw{w}_2$ for the second component $y$ of the pair $(x,y)\equiv \maxw{w}$.} 
$\f{}(w,x) \in \TrS{}{}(x,\maxw{w}_2)$ for all $\Tuple{w,x}\in\dom{f}$. 
%
An \emph{environment strategy} is a \emph{left total}\footnote{Due to the serial assumption on $\G$ it is possible to assume left total environment strategies.} function $\g{} : (\X{}\times \Y{})^+ \rightarrow \X{}$ such that
$\g{}(w) \in \TrE{}{}(\maxw{w})$ for all $w\in (\X{}\times\Y{})^+$. 
We denote the sets of system and environment strategies over $\G$ by $\StrategyS(\G)$ and $\StrategyE(\G)$, respectively.
%
A play $\play{}\in\Play{}$ with $\play{}(k)=\Tuple{x(k),y(k)}$ for all $k\in\Nb$ 
is \emph{compliant} with $\f{}\in\StrategyS(\G)$, $\g{}\in\StrategyE(\G)$ and $\I{}=\Tuple{\X{}\times\Y{}}^+\subseteq\Play{}$
if
there is an $n\in\dom{\play{}}$ such that $\play{}|_{[0,n]} \in \I{}$ and
for all 
 $k\in\dom{\play{}},~k>n$, we have
  \begin{equation}\label{equ:newcompliant}
  x(k)=\g{}\Tuple{\play{}|_{[0,k-1]}}\quad\text{and}\quad
  y(k)=\f{}\Tuple{\play{}|_{[0,k-1]},x(k)}.
\end{equation}
  The set of plays compliant with $\f{}$, $\g{}$ and $\I{}$ is denoted by $\Compl(\f{},\g{},\I{})$ and we define $\Compl(\f{},\I{}):=\bigcup_{\g{}\in\StrategyE(\G)}\Compl(\f{},\g{},\I{})$.

A system strategy $\f{}\in\StrategyS(\G)$ is \emph{winning} for $\Tuple{\G,\I{},\varphi}$ against $\g{}\in\StrategyE(\G)$, if 
\begin{equation}\label{equ:WinStrat}
\AllQSplit{\play{}\in\Compl(\f{},\g{},\I{})}{\ExQ{\xi\in\Play}{\play{}\sconc\xi\in\INTERSECT{\Compl(\f{},\g{},\I{})}{\WinPlays\Tuple{\G,\I{},\varphi}}}.}
\end{equation}
The set of winning strategies for $\Tuple{\G,\I{},\varphi}$ against $\g{}\in\StrategyE(\G)$ is denoted by $\WinStrat(\G,\I{},\varphi,\g{})$ and we define $\WinStrat(\G,\I{},\varphi)=\allowbreak\bigcup_{\g{}\in\StrategyE(\G)}\allowbreak\WinStrat(\G,\I{},\varphi,\g{})$. 

A system strategy $\f{}$ is \emph{dominated} by a system strategy $\f{\prime}$ in the game $\Tuple{\G,\I{},\varphi}$ (see \citet[Def.3]{BrenguierRaskinSassolas_2014}), if for all $\g{}\in\StrategyE(\G)$ holds
\begin{equation*}
\propImp{\f{}\in\WinStrat(\G,\I{},\varphi,\g{})}
 {\f{\prime}\in\WinStrat(\G,\I{},\varphi,\g{})}.
\end{equation*}

A system strategy which is not dominated is called \emph{admissible}. The set of admissible strategies in the play $\Tuple{\G,\I{},\varphi}$ is denoted by $\AdmStrat\Tuple{\G,\I{},\varphi}$.

\paragraph{The Synthesis Problem}
The (unconstrained) synthesis problem takes as input a game $(G, \I{}, \varphi)$ and asks
if there is a winning system strategy for the game.
In most applications, the requirement that the system wins against any adversarial environment strategy is
too stringent.
The constrained synthesis problem additionally takes as input an assumption that models \enquote{likely}
behaviors of the environment as a set of plays $\eass\subseteq \Play{}$.
Intuitively, the constrained synthesis problem asks if the system can win provided that the environment player
is restricted to play strategies that ensure $\eass$.
In the presence of environment assumptions, the synthesis problem looks for \emph{assume-admissible winning strategies} for 
the system (see \cite{BrenguierRaskinSankur_ArXiv_2015} for a discussion why this is an appropriate notion).

By swapping the roles of system and environment we can equivalently define winning and admissible strategies for the environment 
in the game $\Tuple{\G,\I{},\eass}$ as before. 
%
Then a system strategy $\f{}$ is \emph{assume-admissibly winning} for $\Tuple{\G,\I{},\varphi}$ w.r.t.\ $\eass$ (\citet{BrenguierRaskinSankur_ArXiv_2015}, Rule \textsf{AA}) if 
\begin{align}
 &\f{}\in\AdmStrat\Tuple{\G,\I{},\varphi}\quad\text{and}\notag\\
 &\AllQ{\g{}\in\AdmStrat\Tuple{\G,\I{},\eass}}{f\in\WinStrat(\G,\I{},\varphi,\g{})}.\label{equ:AdmStrat}
\end{align}
It should be noted that every winning strategy is assume-admissibly winning w.r.t.\ any assumption, but not vice-versa.

\subsection{Example}\label{sec:Example}

\begin{figure}
\begin{center}
  \input{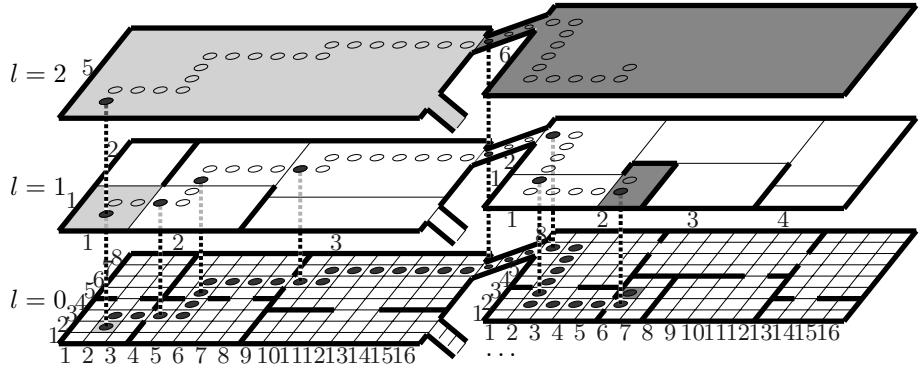}
\end{center}
\caption{Floor plan of the $5$th and $6$th floor of a six-story building. Using the depicted coordinates, we denote by $q_{ij}^k$ and $r_{ij}^k$, respectively, the cell and the room in the $i$th column and $j$th row of floor $k$. Furthermore, $s_{ij},~i<j$ denotes the stair case from floor $f^i$ to floor $f^j$. The workspace of this building is partitioned into grid cells (bottom), rooms (middle) and floors (top) which serve as abstraction layers $l=0$ to $l=2$ as discussed in \REFsec{sec:Example}. The line of dots depicts a path of the robot from the initial state (light gray) to the final state (dark gray) in every layer. Filled circles denote projected states while non-filled circles denote abstract (but not projected) states, as discussed in \REFexp{exp:projection}-\ref{exp:projection2}. }
\label{fig:layers}
\end{figure}

To illustrate the theoretical results and their accompanying assumptions in this paper, we consider a robot that moves in a six story building with known floor plan, depicted in \REFfig{fig:layers} (bottom) for floors $5$ and $6$.

To model this problem as a two-player game graph $\G$, we 
partition the workspace into small cells which form a uniform grid. The resulting grid cells are enumerated by an index set $Q$. 
By assuming that the robot can only be in one grid cell at a time, the system state set is given by $\Y{}=Q$. We furthermore define the set of environment states by $\X{}=\twoup{Q}$, where a state $x\in\X{}$ is a \emph{set} containing all grid cells which are currently occupied by an obstacle. 

This modeling formalism implies that each grid cell in \REFfig{fig:layers} (bottom) represents a system state. 
We model additional properties by adding other binary variables.
For example, by adding a predicate \texttt{Bottle} to the system state, we model whether the robot is carrying a bottle or not. 
As this additional variable might be true in any grid cell, the resulting system state set would consist 
of two copies of the grid world in \REFfig{fig:layers} (bottom), 
where one is annotated with \texttt{Bottle} and the other one is not. 
To keep notation simple, such additional predicates are mostly neglected in this example.

The system transition map $\TrS{}{}$ in $\G$ results from applying an appropriate abstraction method for 
continuous dynamics, e.g., \cite{TabuadaBook}, while adding the obvious restrictions that
\begin{inparaenum}[(i)]
 \item the robot cannot move into an obstacle-occupied cell, and
 \item the robot can only move to adjacent cells that are not separated by a wall.
\end{inparaenum}
For the environment transition map $\TrE{}{}$ several levels of detail can be used to model the movement and 
(dis)appearance of obstacles, see e.g., \cite{WongFinucaneKressGazit_2013,VasileBelta_2014} for examples.

Now consider a task for the robot which asks it to reach a specific room on a specific floor.
This corresponds to a \emph{reachability} winning condition.
In our setting, the winning condition is captured by the language of all plays $\play{}$
such that there exists $k \geq 0$ with $\play{}(k) = (x(k), y(k))$ and $y(k)$ is a cell in the specified room. 
(It can easily be described in linear temporal logic as well.)
The synthesis problem for this specification over the game graph $\G$ finds a strategy (a controller for the robot) that ensures that the robot
eventually reaches the room.  

There are two challenges in applying reactive synthesis in this scenario.
First, the requirement that the robot must reach the room against all possible environments is too stringent.
In such a robot motion example the environment player naturally has a very rich set of possible moves. 
For the specification considered above, the environment can simply keep a couple of doors closed forever to prevent 
the robot to reach its goal. 
However, this adversarial behavior is very unlikely in a real world application as, e.g., employees in 
an office building will always eventually visit/exit their office. 
This is the reason why we introduce environment assumptions that constrain the problem.
A natural environment assumption allowing to realize the above specification models that all staircases are 
always eventually unblocked, all doors get always eventually re-opened, and moving obstacles always eventually allow a passage to exit a room. 

As discussed in \cite{ BrenguierRaskinSassolas_2014}, one cannot simply perform reactive synthesis w.r.t.\ environment assumptions
by considering the implication $\eass \Rightarrow \varphi$ that requires the controller to ensure
$\varphi$ holds only on plays satisfying $\eass$.
This is because the robot may win the game by simply violating the environment assumption (for example,
by blocking a door and preventing the environment from opening it).
Thus, we consider assume-admissible strategies in this paper.

The second challenge is that of scalability. In any realistic model of our problem, the number of states
is so large that existing reactive synthesis tools do not scale.
Our main contribution in this paper is to scale up reactive synthesis techniques by considering \emph{local} structure.
We now consider this in more detail.

As depicted in \REFfig{fig:layers}, there is a natural hierarchy on the states of the workspace imposed by rooms and floors.
That is, the workspace can also be partitioned using the set of rooms $R$ or the set of floors $F$ as 
index sets.\footnote{For simplicity we model the stairs as a separate room and always \enquote{attach} the downward stairs to the respective floor.} 
This partition introduces two abstraction layers with decreasing precision with system state sets $\Y{1}=R$ and $\Y{2}=F$. 
The set of environment states in layers $1$ and $2$ are defined as the set of closed doors 
$\X{1}=\twoup{D}$ and the set of blocked staircases $\X{2}=\twoup{S}$, respectively.
Even though the three layers in \REFfig{fig:layers} are constructed separately, there is a natural abstraction
relation between system states $f\in F$, $r\in R$, and $q\in Q$. 
A system state $q$ is obviously related to the system state $r$ if the grid cell $q$ is \enquote{inside} room $r$.
Furthermore, a door $d$ is marked as \texttt{closed} if all cells intersecting with this door are occupied 
by an obstacle (usually being the door itself in this case), inducing a relation between environment 
states of layers $0$ and $1$. 
In Section~\ref{sec:HGG}, we present \emph{abstract game graphs} (AGGs) which capture such hierarchies in reactive games.
 
%
%

The abstraction relations naturally decompose every layer in the example into small, local game graphs located \enquote{inside} 
a higher level system state: the game graph $\G$ is decomposed in local game graphs $\G_{r},~r\in R$. This is possible for this example as 
the set of possible moves in one room is independent from the part of the environment state that does not belong to this 
context, e.g., all the obstacles contained in the set $x$ that are not located inside this room.
In Section~\ref{sec:LGG}, we introduce {\em local game graphs} (LGGs) which decompose AGGs to model this locality within the hierarchy.
 
To exploit this local structure in reactive synthesis, we additionally require that the specification is also given as a set of local specifications,
one for each local game; otherwise, there is no obvious way to automatically break a global specification into local synthesis problems. 
For example, for the reachability task, one can consider a specification of reaching a room at the higher layer, and reaching from one point
of a room to a prescribed exit point in the lower layer.
Correspondingly, notice that the environment assumptions can also be decomposed into layers.

As a second example, consider the more complex task:
\begin{center}
 \textit{\enquote{Collect all empty bottles in the building and return\\ them to the kitchen in the $5$th floor.}}
\end{center}
This task can be manually decomposed in a natural fashion as follows.
The level $2$ task asks the robot to visit all floors of the building and to return to floor $5$ whenever its capacity to carry empty bottles is reached. 
While in one floor, the level $1$ task asks the robot to visit all rooms until the carrying capacity is reached, and to visit 
the kitchen whenever the latter is true and the robot is in floor $5$. 
Finally, the level $0$ tasks ask the robot to search for empty bottles in a single room, approach each bottle and pick it up.
In this paper we assume that both the system specification and the environment assumptions are already given 
in a decomposed manner.
The automatic decomposition of a global winning condition into local ones is an orthogonal, difficult, problem.

In Section~\ref{subsec:DefHierarchicalGames}, we define \emph{hierarchical reactive games} (HRGs) by combining the set of LLGs over hierarchies with a set of local winning conditions and a set of local environment assumptions. This generates a set of local games over an LGG w.r.t. a local specification $\varphi$ and a local assumption $\eass$.

The main challenge for reactive synthesis for HRGs is that the games played at the various layers interact.
That is, a strategy at a higher layer (``go to the kitchen'') introduces additional constraints at the
lower layer (``the higher level strategy requires that the robot should go to the exit that takes it to the
kitchen'').
In Section~\ref{sec:Strategies}, we provide a synthesis algorithm that computes a dynamic controller
for HRGs.
The controller computes assume-admissible strategies for each local game, 
and dynamically updates the winning conditions and strategies through the hierarchy.
We prove that the algorithm is sound and that it aborts the game only when a local subgame cannot be won by the system
against admissible strategies of the environment.

\section{Hierarchical Decomposition}\label{sec:HGG}

We now introduce a hierarchy of $L$ two player game graphs where the higher layers are 
a more abstract representation of the original game graph at layer $l=0$.

\subsection{Layering, Abstract Plays, and Timescales}\label{sec:R}

%
Let $\G=\Tuple{\X{},\Y{},\TrE{}{},\TrS{}{}}$ be a game graph.
A sequence $\tuple{\X{0}, \Y{0}}, \tuple{\X{1},\Y{1}}, \ldots, \tuple{\X{\lmax}, \Y{\lmax}}$
is a \emph{layering} of $\G$ if 
\begin{inparaenum}[(i)]
\item $\X{0} = \X{}$ and $\Y{0} = \Y{}$, and
\item for each $l\in [1,\lmax]$, there exist \emph{abstraction functions}
$\Ry{l} : \Y{l-1}\fun \Y{l}$ and
$\Rx{l} : \BR{\X{l-1}\times\Y{l-1}} \fun \X{l}$.
\end{inparaenum}


Notice that while the system abstraction function maps system states at level $l-1$ to system states at level $l$,
the environment abstraction function $\Rx{l}$ maps a pair $(x,y)$ of environment and 
system states at level $l-1$ into an environment state at level $l$. 
This allows us to incorporate the loss of direct control with increasing abstraction level, 
as illustrated in the following example.

\begin{example}
Consider the robot in \REFsec{sec:Example} and assume that the system states of layer $0$ are extended by the binary variable \texttt{Bottle}, resulting in the state $\Set{q,\mathtt{Bottle}}$ if the robot is in cell $q$ and carries a bottle and the state $\Set{q}$ if the latter is not true. In this example, a transition from state $\Set{q}$ to $\Set{q,\mathtt{Bottle}}$ is enforceable in layer $0$ if there is a bottle in cell $q$ (which can be modeled by a corresponding environment variable) assuming that the robot can always pick up a bottle when it is in this cell

Now assume that the specification in the room level asks the robot to go to the kitchen, if it is carrying a bottle. To realize this task, a strategy in layer $1$ does not need to \emph{enforce} the robot to pick up a bottle in a particular room (because it might not actually know in which rooms bottles are located) but only \emph{observe} that the latter happened. This intuition can only be modeled if $\mathtt{Bottle}$ is included in the environment states  rather than the system states of layer $1$. To be able to trigger this environment variable in layer $1$ when the robot picks up a bottle, the tuple 
$\Tuple{x,\Set{q,\mathtt{Bottle}}}\in\X{0}\times\Y{0}$ must be projected to an environment state $\Set{\mathtt{Bottle}}\cup x'\in\X{l}$ using the map $\Rx{1}$.
%
%
%
\end{example}

For notational convenience, we define 
the composition of abstraction functions $\Rxup{l}:\BR{\X{}\times\Y{}}\fun\X{l}$ and $\Ryup{l}:\Y{}\fun\Y{l}$ as
\begin{subequations}\label{equ:layers}
 \begin{align}
 &\AllQ{x\in\X{},y\in\Y{}}{\Rxup{l}(x,y)=\Rx{l}\BR{\Rx{l-1}\BR{\hdots\Rx{1}\BR{x,y}}}},\label{equ:Rup:x}\\
  &\AllQ{y\in\Y{}}{\Ryup{l}(y)=\Ry{l}\BR{\Ry{l-1}\BR{\hdots\Ry{1}\BR{y}}}}\label{equ:Rup:y}
\end{align}
\end{subequations}
and the special cases $x=\Rxup{0}(x,y)$ and $y=\Ryup{0}(y)$.

A layering induces an abstraction for a play $\pi\in\Play{}$ for each layer $l>0$ as follows.
Given a game $\G{}$, a play $\play{}\in\Play{}$, and layers $\tuple{\X{l},\Y{l}}_{l=0}^\lmax$ with
abstraction functions $\Rx{l}$ and $\Ry{l}$, we define the set of \emph{abstract plays} $\Pi=\Set{\pi^l}_{l=0}^L$ of $\pi$ by 
$\pi^l\in(\X{l}\times\Y{l})^\infty$ with $\pi^l(k)=\Tuple{x^l(k),y^l(k)}$ s.t.\
\begin{equation}\label{equ:pil}
  \AllQ{k\in\domp{\pi}}{
\begin{propConjA}
   x^l(k)=\Rxup{l}\BR{x(k),y(k-1)}\\
   y^l(k)=\Ryup{l}(y(k))
\end{propConjA}
}
\end{equation}
and $\pi^l(0)=\Tuple{\Rxup{l}(x(0), y(0)),\Ryup{l}(y(0))}$.

%
  
Intuitively, the abstract plays in $\Pi$ are an abstraction of the 
play $\pi$ which becomes coarser the higher the layer, as 
multiple system and environment states are clustered into one state in a higher level. 
Specifically, this implies that state changes occur less frequently in a higher 
level than in the play $\pi$ as outlined in the following example.

\begin{example}\label{exp:projection}
Consider the path of the robot depicted by filled cycles in \REFfig{fig:layers} (bottom). 
This path represents the system state component $y$ of a play $\pi\in\Play{}$. 
Applying the second line of \eqref{equ:pil}, this sequence $y$ can be abstracted to layer $l=1$ and $l=2$ as follows.
\begin{align*}
 \begin{matrix}
  y=&q_{22}^5&q_{23}^5&q_{33}^5&q_{43}^5&q_{53}^5&q_{54}^5&q_{55}^5&q_{56}^5&\hdots\\[0.1cm]
  y^1=&r_{11}^5&r_{11}^5&r_{11}^5&r_{21}^5&r_{21}^5&r_{21}^5&r_{22}^5&r_{22}^5&\hdots\\[0.1cm]
   y^2=&f^5&f^5&f^5&f^5&f^5&f^5&f^5&f^5&\hdots
 \end{matrix}
\end{align*}
The abstract sequences $y^1$ and $y^2$ are depicted in \REFfig{fig:layers} (middle) and (top), respectively. 
The state changes in levels $1$ and $2$ correspond to changes in rooms and floors, respectively.
While the state at level $0$ changes in each time step,
observe that state transitions in layers $1$ and $2$ only happen irregularly and 
not at every time point. 
It should be noted that environment states in layer $1$ and $2$, 
i.e., the set of closed doors and blocked stairs, can change independently from system state changes 
and is not illustrated in \REFfig{fig:layers}. 
\end{example}
  
\REFexp{exp:projection} illustrates that an abstract play $\pi^l$ is usually not turn-based. 
To obtain a turn-based game and to remove redundant information, we introduce a new time 
scale for every layer which is triggered by changes in the system states in an abstract game $\pi^l$ as follows.
%
Given a play $\play{}\in\Play{}$ and a layer $l\in [0,\lmax]$, the \emph{timescale transformation} $\kappa^l$ of $\play{}$ in layer $l$ 
is the identity function if $l = 0$,
and defined by the strictly monotone sequence $\kappa^l\in\Nb^\infty$ s.t.
\begin{subequations}\label{equ:kappa}\allowdisplaybreaks
 \begin{align}
 &\kappa^l(0)=0,\\
  &\AllQSplit{m\in\dom{\kappa},m>0,k\in[\kappa(m-1),\kappa(m))}{
  y^l(k)=y^l(\kappa^l(m-1))\neq y^l(\kappa^l(m))
  } \\
  &\text{and}\quad \AllQ{k>\maxw{{\kappa^l}}}{y^l(k)=y^l(\maxw{\kappa^l})},
 \end{align}
\end{subequations}
 otherwise.
The set of \emph{projected plays} $\PlP=\Set{\plP{l}}_{l=0}^\lmax$ of $\play{}$ with $\plP{l}=\Tuple{\plPx{l},\plPy{l}}$ is defined as the sub-sequence of the abstract play $\pi^{l}$ at time points given by
$\kappa^l$ for every $l\in[1,L]$. Formally, 
\begin{equation}\label{equ:projpi}
 \AllQ{k\in\dom{\kappa^l}}{\plP{l}(k)=\pi^l(\kappa^l(k))}.
\end{equation}

A projected play $\plP{}$ is called \emph{infinite} if $|\plP{}| = \infty$ and \emph{finite} otherwise. 
While plays $\pi\in\Play{}$ can always be made infinite (by the serial assumption on the transition relations), its projection $\plP{l}$ to layer $l >0$ need not be infinite.
For example, if the robot from \REFsec{sec:Example} should just move within room $r^5_{11}$,
this obviously induces an infinite play $\pi$. 
However, its projection to the room layer is given by $\plP{1}=r^5_{11}$,
i.e., $\plP{1}$ is finite with length $1$. 
%
%


\begin{example}\label{exp:projection2}
 Consider the abstract sequences $y^1$ and $y^2$ in \REFexp{exp:projection}. Using \eqref{equ:kappa} and \eqref{equ:projpi} their induced time scale transformations are given by
 \begin{equation*}
  \kappa^1=0~3~6~\hdots\quad\text{and}\quad \kappa^2=0~20\\
 \end{equation*}
 and the resulting projections for layer $1$ and $2$ are given by
 \begin{align*}
  \plPy{1}=r_{11}^5~r_{12}^5~r_{22}^5\hdots\quad\text{and}\quad\plPy{2}=f^5~f^6
 \end{align*}
corresponding to changes in rooms and floors respectively at those times.
In \REFfig{fig:layers}, system states of projected plays are depicted by filled circles, whereas states only belonging to abstract plays are depicted by non-filled cycles. 
\end{example}

It can be easily shown (see \REFlem{lem:kappallp1} in \REFapp{sec:app:proofs}) that the range of the timescale transformation
$\kappa^{l+1}$ is a subset of the range of $\kappa^{l}$; if there is an event
at the $(l+1)$st layer, there is a corresponding event at the $l$th (and so, in each lower)
layer. Using this observation we can simplify notation by defining
\begin{equation}\label{equ:kappallp1}
 \kappa^{l+1}_l(k):=\BR{\kappa^l}^{-1}\BR{\kappa^{l+1}(k)}
\end{equation}
to denote the position in the $l$th layer of the 
$k$th event in the $(l+1)$st layer.

\subsection{Abstract Game Graphs}\label{subsec:AGG}
Using the notion of abstract states and plays from the previous section, we now construct game graphs for every layer $l$.
We remark that the actual game is only played in 
the lowest layer, i.e., in the game graph $\G$, and the higher layers only model
projected plays of this game.
%
\begin{definition}\label{def:Gl}
Let $\G=\Tuple{\X{},\Y{},\TrE{}{},\TrS{}{}}$ be a game graph, and 
$\tuple{\X{l},\Y{l}}_{l=0}^\lmax$ a layering of $\G$ using the abstraction functions $\Rx{l}$ and $\Ry{l}$.
Then we define the set of \emph{abstract game graphs} (AGG) $\Set{\G^l}_{l=0}^{L}$ for each layer $l\in[1,\lmax]$ by $\G^l:=\Tuple{\X{l},\Y{l},\TrE{l}{},\TrS{l}{}}$ s.t.
{\allowdisplaybreaks
\begin{subequations}\label{equ:Gl}
\begin{align}
 &\propAequ{
   x'\in\TrE{l}{}\BR{x,y}
}{  \ExQ*{\play{}\in\Play{},y'\in\Y{l}}{
\begin{propConjA}
 \play{l}(\kappa^l(0))=\Tuple{x,y}\\
 \ExQ{k\in(0,\kappa^l(1)]}{\play{l}(k)=\Tuple{x',y'}}
\end{propConjA}
}
}\label{equ:Gl:TrE}\\
 &\propAequ{
   y'\in\TrS{l}{}\BR{x,y}}{
  \ExQ*{\play{}\in\Play{},x'\in\X{l}}{
   \begin{propConjA}
  \play{l}(\kappa^l(1)-1)=\Tuple{x',y}\\
  \play{l}(\kappa^l(1))=\Tuple{x,y'}
  \end{propConjA}
 }
 }.\label{equ:Gl:TrS}
\end{align}
\end{subequations}} 
and for $l=0$ by $\G^0:=\G$.
\end{definition}

Intuitively, the maps $\TrE{l}{}$ and $\TrS{l}{}$ collect all transitions that can occur in projected plays $\plP{l}$ of possible lowest level plays $\pi\in\Play{}$, as illustrated in the following example. It should be noted that all lowest level plays $\pi$ are existentially quantified in \eqref{equ:Gl}, i.e., all possible plays in the lowest layer are considered.

\begin{example}\label{exp:Gl}
 Consider the play $\pi\in\Play{}$ and its abstract play $\pi^1$ depicted in \REFfig{fig:absttrans}. The existence of the play $\pi$ introduces the depicted system and environment transitions using \eqref{equ:Gl:TrE} and \eqref{equ:Gl:TrS}, respectively. Observe that the construction considers every environment change (induced by the play $\pi$) as an environment transition from the environment state at the last triggering instance indicated by $\kappa$. Furthermore, system transitions are only generated at triggering times. It can be seen in \REFfig{fig:absttrans} that the environment state in layer $l>0$ possibly changes multiple times before a system state change follows. 
\end{example}
\begin{figure}[t!]
\begin{center}
 \begin{tikzpicture}[node distance=1cm]
\def\h{0} \def\ha{1.5} 
\def\v{0}\def\va{0.7} 

\foreach \x in {0,...,2} { 
\node (no\x) at (\h+2*\x*\ha,0) {$\Tuple{x_{\x},y_{\x}}$};
\draw [-latex] (no\x) -- node[yshift=0.2cm] () {$\TrE{}{}$} ++ (0.6*\ha,0);
\draw [->,dashed,black!50,thick] (no\x) -- ++ (0,\va);
\coordinate[below of=no\x, node distance=0.6cm] (bno\x);
 }
 \foreach \x in {1,...,2} { 
\node (ne\x) at (\h+2*\x*\ha-1*\ha,0) {$\Tuple{x_{\x},y_{\number\numexpr\x-1\relax}}$};
\draw [-latex] (ne\x) -- node[yshift=0.2cm] () {$\TrS{}{}$} ++ (0.6*\ha,0);
\draw [->,dashed,black!50,thick] (ne\x) -- ++ (0,\va);
\coordinate[below of=ne\x, node distance=0.6cm] (bne\x);
 }
\node (ne3) at (\h+4.8*\ha,0) {$\hdots$};

\node at (\h-0.5*\ha,0) {$\pi\hspace{-0.1cm}:$};
\node at (\h-0.5*\ha,1cm) {$\pi^1\hspace{-0.1cm}:$};

\node[above of=no0] (a0) {$\Tuple{x^1_0,y^1_0}$};
\node[above of=ne1] (a1) {$\Tuple{x^1_0,y^1_0}$};
\node[above of=no1] (a2) {$\Tuple{x^1_0,y^1_0}$};
\node[above of=ne2] (a3) {$\Tuple{x^2_0,y^1_0}$};
\node[above of=no2] (a4) {$\Tuple{x^2_0,y^1_1}$};
\node[above of=ne3] (a5) {$\hdots$};

\draw [-latex] (a0.north) to[out=20,in=160] node[yshift=0.25cm,pos=0.9] {$\TrE{1}{}$} (a1.north);
\draw [-latex,thick] (a0.north) to[out=30,in=150] node[yshift=0.25cm,pos=0.7] {$\TrE{1}{}$} (a3.north);
\draw [-latex,thick] (a3.north) to[out=20,in=160] node[yshift=0.25cm] {$\TrS{1}{}$} (a4.north);

\draw [->,line width=1.5pt] (bno0) -- ++(4.8*\ha,0);
\draw [line width=1pt] (bno0)+(0,0.1) -- node[anchor=west,yshift=-0.3cm,xshift=-0.2cm] {$0=\kappa^1(0)$} ++(0,-0.1);
\draw [line width=1pt] (bno1)+(0,0.1) -- node[yshift=-0.3cm] {$1$} ++(0,-0.1);
\draw [line width=1pt] (bno2)+(0,0.1) -- node[anchor=west,yshift=-0.3cm,xshift=-0.2cm] {$2=\kappa^1(1)$} ++(0,-0.1);
\end{tikzpicture}
\end{center}
\vspace{-0.6cm}
\caption{Generation of system and environment transitions for layer $l=1$ from a play $\pi$ as formalized in \REFdef{def:Gl} and discussed in \REFexp{exp:Gl}.}
\label{fig:absttrans}
\end{figure}
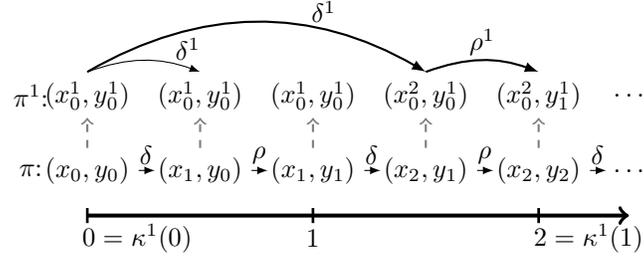

The construction in \REFdef{def:Gl} allows us to prove that 
projected plays $\plP{l}$ as defined in \eqref{equ:projpi} are 
also plays in the game graph $\G^{l}$, i.e., $\plP{l}\in\Play{l}$. 
Intuitively, the proof shows that there always exist transitions, as the 
ones emphasized in \REFfig{fig:absttrans}, connecting system and environment states at triggering times.

\begin{proposition}
\label{prop:Play-l}
For any game $\G$, any play $\play{}\in\Play{}$, and any $l\in[0,\lmax]$, we have that
$\plP{l}$ is a play in $\G^l$, i.e., $\plP{l}\in\Play{l}$.
\end{proposition}

 \begin{proof}
 The claim follows directly from \REFlem{lem:playGl} in \REFapp{sec:app:proofs} as \eqref{equ:playp_def:b} 
holds for $\plP{l}$ and $\G^l$ when we pick $n=\kappa^l(m+1)$ in \eqref{equ:lem:playGl}.
\end{proof}

\section{Context-Based Decomposition}\label{sec:LGG}

A set of AGGs imposes an abstraction hierarchy on top of a given game graph $\G$. 
However, AGGs by themselves are not enough to decompose a synthesis problem.
For example, if the winning condition is given by a set of plays 
on the lowest layer, the induced abstraction layers cannot be exploited by a synthesis algorithm. 
In order to derive an efficient synthesis technique, in this section, we introduce the second
ingredient: {\em local} winning conditions, which induce {\em local game graphs}.

Roughly, a \emph{local} winning condition for the game $\G^l$ at layer $l$ is a 
set of abstract plays $\play{l}$ whose states belong to a single state at layer $l+1$.
For example, reaching a different floor is a local specification at layer 2.
A synthesis procedure to enforce $\varphi^\lmax$ would require solving games at lower
levels; in our example, the robot will have to successively reach a set of rooms, followed by the stairs
to achieve its goal.
Each of these ``lower level'' games occur in, roughly, the ``local''
game structure defined by states in the lower level that map to the current
state of the higher level.
We formalize this notion as \emph{local game graphs}.

\subsection{Local Game Graphs over Hierarchies}

Fix a layer $l$ and consider the games $\G^l$ and $\G^{l+1}$.
Consider a system state $\nu \in \Y{l+1}$.
A first attempt to define a local game is to restrict the game $\G^l$
to the set of system states $\set{y\in \Y{l} \mid \Ry{l+1}(y) = \nu}$. 
However, this is not sufficient, because
plays in the local game should be allowed to leave 
the region specified by $\nu$ for one step at the end. 
This is necessary to ensure that plays in consecutive local games can be concatenated to form
a play over the game graph $\G^l$ without formalizing a special reset action, as e.g., used 
in modular games by \cite{AlurLaTorreMadhusudan_2003b}. 
To account for these states, we introduce the $\ON{Post}$ operation:
%
  \begin{equation}
  \ON{Post}^l(\nu):=\SetCompX{\nu'\in\Y{l}}{
  \begin{propConjA}
   \nu'\neq\nu\\
  \ExQ{x\in\X{l}}{\nu'\in\TrS{l}{}(x,\nu)} 
  \end{propConjA}
  }.
 \end{equation}
Including the one-step post states allows us to view the actual game as a layer $0$ game and use the hierarchical and local decompositions as modeling formalism for hierarchical controller synthesis only.

Considering environment states instead of system states, a straightforward restriction to a 
context $\nu$ is not naturally given by $\Rxup{l+1}$, as the following example shows. 

\begin{figure*}
\begin{center}
  \input{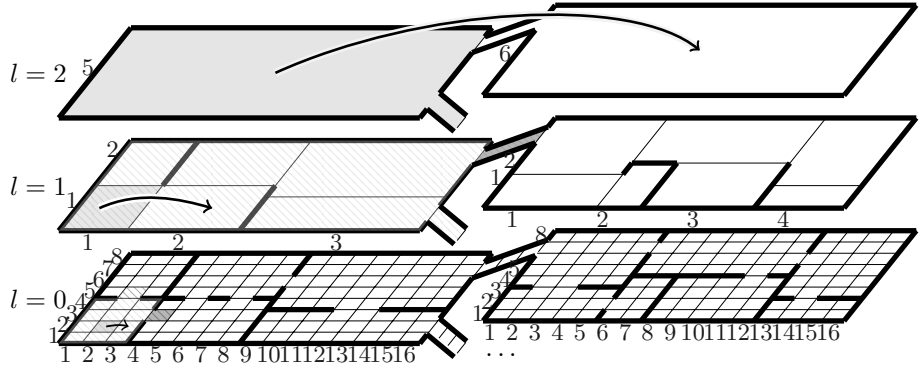}
\end{center}
\vspace{-0.5cm}
\caption{Floor plan from \REFfig{fig:layers}. The striped areas in layers $0$ and $1$ correspond to $\Y{0}_{r^5_{11}}$ and $\Y{1}_{f^5}$, respectively. The three arrows denote context changes requested by layer $l$ which induce a reachability specification for layer $l-1$ whose initial and goal states are depicted in light and dark gray, respectively.}
\label{fig:contexts}
\end{figure*}

\begin{example}\label{exp:rxone}
Consider the example from \REFsec{sec:Example} and its floor plan depicted in \REFfig{fig:contexts}. 
Recall from \REFsec{sec:Example} that an environment state $x\in\X{0}$ contains all grid cells that are occupied by an obstacle. 
However, by playing a game in room $r^5_{11}$ one is only interested in obstacles that are located 
inside $\Y{0}_{r^5_{11}}$. 
\end{example}

Therefore, instead of using $\Rxup{l+1}$ to restrict $\X{l}$ to context $\nu$, we use a restricting function $\rx{l}{\nu}$. 
For \REFexp{exp:rxone}, the map $\rx{1}{r^5_{11}}$ simply maps the set $x$ of obstacle locations to the subset 
$x'\subseteq x$ of such locations that are inside the striped area in layer $0$ of \REFfig{fig:contexts}. 
For notation convenience, we define $\rx{L}{}$ as the identity map.
Using the above intuition, we define \emph{local game graphs} as follows.

\begin{definition}\label{def:Gly}
Given an AGG $\G^l$, the \emph{local} game graph (LGG)
$\G^{l}_{\nu}:=\Tuple{\X{l}_{\nu},\Y{l}_{\nu},\TrE{l}{\nu},\TrS{l}{\nu}}$ 
at layer $l$ restricted to $\nu\in\Y{l+1}$ consists of 
\begin{subequations}\label{equ:Gly}\allowdisplaybreaks
 \begin{align}
&\X{l}_{\nu}:=\SetCompX{\rx{l}{\nu}(x)}{x\in\X{l}}~\text{and}\label{equ:Gly:X}\\
&\Y{l}_{\nu}=\Yli{l}{\nu}\cup\Yla{l}{\nu}\label{equ:Gly:Y}\\
\SUCHTHAT~&\Yli{l}{\nu}:=\SetComp{y\in\Y{l}}{\nu=\Ry{l+1}(y)}~\text{and}\label{equ:Gly:Yli}\\
&\Yla{l}{\nu}:=\SetCompX{y'\in\Yli{l}{\nu'}}{
\begin{propConjA}
\nu'\in\ON{Post}^l(\nu)\\
\ExQ{y\in\Yli{l}{\nu},x\in\X{l}_{\nu}}{y'\in\TrS{l}{}(x,y)}
\end{propConjA}},\label{equ:Yla}
\end{align}
\end{subequations}
and transition maps 
$\TrE{l}{\nu}:\X{l}_{\nu}\times\Yli{l}{\nu}\fun2^{\X{l}_{\nu}}$ 
and $\TrS{l}{\nu}:\X{l}_{\nu}\times\Yli{l}{\nu}\fun2^{\Y{l}_{\nu}}$ defined as: 
\begin{subequations}\label{equ:Gly:Tr}\allowdisplaybreaks
\begin{align}
 &\propImp{\propConj*{
  x'\in\TrE{l}{}(x,y)}{
  y\in\Yli{l}{\nu}}
}{\rx{l}{\nu}(x')\in\TrE{l}{\nu}(\rx{l}{\nu}(x),y)}
\quad\text{and}
\label{equ:Gly:TrE}\\
 &\propImp{\propConj*{
 y'\in\TrS{l}{}(x,y)}{
 y\in\Yli{l}{\nu}\wedge y'\in \Y{l}_{\nu}}
 }{y'\in\TrS{l}{\nu}(\rx{l}{\nu}(x),y)}.
\label{equ:Gly:TrS}
\end{align}
\end{subequations}
We write
$[\Glall]:=\left\{\left\{\G^l_\nu\right\}_{\nu\in\Y{l+1}}\right\}_{l=0}^{\lmax-1}\cup\Set{\G^L}$
for the set of LGGs over $\G$. 
\end{definition}
\begin{example}\label{exp:rx}
 Consider the example from \REFsec{sec:Example} and its floor plan depicted in \REFfig{fig:contexts}. The striped areas in layers $0$ and $1$ correspond to the context restricted system state sets $\Y{0}_{r^5_{11}}$ and $\Y{1}_{f^5}$, respectively. It is easy to see that
  $\Yla{0}{r^5_{11}}=\Set{q^5_{25},q^5_{43}}$ and
  $\Yla{1}{f^5}=\Set{s_{56}}$,
 while layer $l=2$ is not decomposed.
\end{example}


In the robot example of \REFsec{sec:Example} the generated set of LGGs is \enquote{truly local} in the sense that the local system dynamics do not depend on environment variables from other contexts. 
E.g., an obstacle in another room $r'$ does not influence the dynamics of the robot in room $r\neq r'$. 
This inherent decomposability of the system dynamics, similar to the natural relations among states of different 
layers, is a feature of the system we want to control which is necessary for the subsequently proposed synthesis algorithm and formalized in the following assumption.

\begin{assumption}\label{ass:Xbinlnu}
For every layer $l\in[0,\lmax-1]$ and context $\nu\in\Y{l+1}$ it holds for all $x\in\X{l}$ and $y\in\Yli{l}{\nu}$ that
\begin{equation}\label{equ:truelylocal}
 \propImp{y'\in\TrS{l}{}\Tuple{x,y}}{y'\in\TrS{l}{}\Tuple{\rx{l}{\nu}(x),y}}.
\end{equation}
%
\end{assumption}

It should be noted that the right hand side of \eqref{equ:truelylocal} uses $\TrS{l}{}$ instead of $\TrS{l}{\nu}$. Therefore, $\TrS{l}{\nu}\subseteq\TrS{l}{}$ if \REFass{ass:Xbinlnu} holds, which implies that in this case \eqref{equ:Gly:Tr} holds in both directions.

Similarly to \REFprop{prop:Play-l} we can prove that the part of a play $\play{l}$ that takes place in context 
$\nu$ is actually a play in $\G^l_\nu$. 
However, to formalize this we need to define \emph{local plays} which are projected to the current context.
Given a set of LGGs $[\Glall]$, a play $\play{}\in\Play{0}$ and its sets of abstract and projected plays $\Pi$ and $\PlP$, the \emph{local restriction} of $\pi^{l}$ and $\plP{l}$ is defined for all $m\in\domp{\plP{l}}$ by 
\begin{subequations}
  \begin{align}
 &\pld{l}(m):=\Tuple{\plxd{l}(m),\ply{l}(m)} &&\text{with}\quad\plxd{l}(m):=\rx{l}{y^{l+1}(\kappa^l(m)-1)}\BR{\plx{l}(m)}\quad\text{and}\label{equ:pld}\\
 &\plPd{l}(m):=\Tuple{\plPxd{l}(m),\plPy{l}(m)} &&\text{with}\quad\plPxd{l}(m):=\rx{l}{y^{l+1}(\kappa^l(m)-1)}\BR{\plPx{l}(m)}.\label{equ:plPd}
 \end{align}
\end{subequations}
The restriction of $\plx{l}(m)$ (resp. $\plPx{l}(m)$) at time $k=\kappa^l(m)$ is defined w.r.t.\ the last system state $y^{l+1}(k-1)$ as $y^{l+1}(k)$ is only available after the next system move that is depended on $x(k)$. 
The local restriction $\plPd{l}$ of the projected play introduces a sequence $\plPdm{l}$ of local projected plays defined by
 \begin{subequations}\label{def:plPdm}
 \begin{align}
  &\AllQ{m\in\domp{\plP{l+1}}}{\plPdm{l}(m-1):=\plPd{l}\ll{\kappa^{l+1}_l(m-1),\kappa^{l+1}_l(m)}}\label{equ:plPdm:a}\\
  &\text{and}\quad\plPdm{l}(\maxk{\plP{l+1}})=\maxw{\plPdm{l}}:=\plPd{l}\ll{\maxw{\kappa^{l+1}_l},\maxk{\plP{l}}},\label{equ:plPdm:b}
 \end{align}
 \end{subequations}
 where $\maxk{w}=|w|-1$ denotes the time of the last element of $w$.
 We write $ \rhoall:=\left\{\plPdm{l}\right\}_{l=0}^{\lmax-1}\cup\Set{\plPdm{L}}$
for the set of all such sequences induced by $\pi$, where $\plPdm{L}(0)=\plP{L}$ and $\maxk{\plPdm{L}}=0$.

\begin{example}\label{exp:plPdm}
 Consider the play $\pi$ whose $y$-component is depicted by filled cycles in \REFfig{fig:layers} (bottom). For illustration purposes, assume a static environment with a closed door between room $r^5_{11}$ and $r^5_{12}$, denoted by the binary variable $d$, and an obstacle in $q^5_{63}$. The closed door, which is an environment variable for layer $1$, corresponds to obstacles in $q^5_{24}$ and $q^5_{25}$ for layer $0$. For this play, the local plays contained in the set $\rhoall$ are given by the following strings.
 \newcommand{\Th}[1]{\Tuple{#1,\cdot}}
 \begin{align*}\allowdisplaybreaks
  \plPdm{0}(0)&=\Tuple{\Set{q^5_{24},q^5_{25}},q^5_{22}}\Tuple{\Set{q^5_{24},q^5_{25}},q^5_{23}}\Tuple{\Set{q^5_{24},q^5_{25}},q^5_{33}}\Tuple{\Set{q^5_{24},q^5_{25}},q^5_{43}}\\
  \plPdm{0}(1)&=\Tuple{\Set{q^5_{24},q^5_{25}},q^5_{43}}\Tuple{\Set{q^5_{63}},q^5_{53}}\Tuple{\Set{q^5_{63}},q^5_{54}}\Tuple{\Set{q^5_{63}},q^5_{55}}\\[-0.3cm]
  &~~\vdots\\[-0.2cm]
  \plPdm{0}(7)&=\Tuple{\Set{\bot},q^6_{62}}\Tuple{\Set{\bot},q^6_{63}}\\[0.2cm]
  \plPdm{1}(0)&=\Tuple{\Set{d},r^5_{11}}\Tuple{\Set{d},r^5_{21}}\Tuple{\Set{d},r^5_{22}}\Tuple{\Set{d},r^5_{32}}\Tuple{\Set{d},s_{56}}\\
  \plPdm{1}(1)&=\Tuple{\Set{d},s_{56}}\Tuple{\Set{\bot},r^6_{12}}\Tuple{\Set{\bot},r^6_{11}}\Tuple{\Set{\bot},r^6_{21}}\\[0.2cm]
  \plPdm{2}(0)&=\Tuple{\Set{\bot},f^{5}}\Tuple{\Set{\bot},f^{6}}.
 \end{align*}
 where $\Set{\bot}$ denotes that no obstacles are present.
Due to the definition of $\Y{l}_\nu$ in \REFdef{def:Gly}, contexts of neighboring cells overlap. This is also visible by the above local plays, which overlap for one time instant. E.g, the state $\Tuple{\Set{q^5_{24},q^5_{25}},q^5_{43}}$ belongs both to $\plPdm{0}(0)$ and $\plPdm{0}(1)$, which are the local plays in context $Y^0_{r^5_{11}}$ and  $Y^0_{r^5_{21}}$, respectively. As we use the convention that the environment moves first, the environment variables of such overlapping states are always restricted to the context, which is currently left.
\end{example}

%

\begin{proposition}\label{prop:playGly_simple}
Let $[\Glall]$ be a set of LGGs and $\Play{l}_{y}$ the set of plays in $\G^l_{y}$. Furthermore, let $\play{}\in\Play{}$ and $\rhoall$ its induced set of local projected play sequences. 
 Then 
it holds for all $l\in[0,\lmax-1]$ and $m\in\dom{\plP{l+1}}$ that
\begin{equation}\label{equ:prop:playGly_simple}
 \plPdm{l}(m)\in\Play{l}_{\plPy{l+1}(m)}.
\end{equation}
\end{proposition}

\begin{proof}
 \eqref{equ:prop:playGly_simple} follows by combining the last lines of \eqref{equ:playGly} and \eqref{equ:playGly:end} in \REFlem{lem:playGly} proven in \REFapp{sec:app:proofs}.
\end{proof}

%

%

\subsection{Hierarchical Reactive Games over Sets of LGGs}\label{subsec:DefHierarchicalGames}

We have seen in the example of \REFsec{sec:Example} that the motivation 
for constructing LGGs comes from the natural decomposability of system dynamics, environment assumptions and tasks into 
local and global components which are naturally restricted to a context $\nu\in\Y{l+1}$. 
Recall that local specifications should intuitively only contain finite strings to eventually allow 
progress in the higher layer upon completion of the local task. 
This observation is formalized as follows. %
%
Given a set $[\Glall]$ of LGGs, layer $l\in[0,\lmax-1]$, and context $\nu\in\Y{l+1}$, the sets
\begin{align}\label{equ:PhiallR}
&\varphi^l_\nu\subseteq\INTERSECT{\Tuple{\X{l}_\nu\times\Yli{l}{\nu}}^*}{\Play{l}_\nu}\quad\text{and}\quad
\eass^l_\nu\subseteq\INTERSECT{\Tuple{\X{l}_\nu\times\Yli{l}{\nu}}^\infty}{\Play{l}_\nu}
\end{align}
are the \emph{local system specification} and the \emph{local environment assumption} for $\G^l_\nu$, respectively.
The sets
$\varphi^L\subseteq\Play{L}$ and $\eass^L\subseteq\Play{L}$
are a system specification and an environment assumption for $\G^L$, respectively.
We define sets of local system specifications and local environment assumptions over $[\Glall]$ as
\begin{align}\label{equ:phiall}
 &\PhiallR:=\left\{\hspace{-0.1cm}\left\{\varphi^l_\nu\right\}_{\nu\in\Y{l+1}}\right\}_{l=0}^{\lmax-1}\cup\Set{\varphi^L}\quad\text{and}\quad
 \Alphaall:=\left\{\hspace{-0.1cm}\left\{\eass^l_\nu\right\}_{\nu\in\Y{l+1}}\right\}_{l=0}^{\lmax-1}\cup\Set{\eass^L}.
\end{align}

A winning strategy for a local specification in layer $l+1$ induces transitions from a state $\Tuple{x,y}$ to a (possibly different) state $\Tuple{x,y'}$. As $y,y'\in\Y{l+1}$ are different contexts for layer $l$, this order of contexts must be obeyed by the strategy in layer $l$. 
Therefore, we need a proper translation of transitions in level $l+1$ into reachability specification for local games in layer $l$ and combine these specifications with the given low level tasks. 
%
%
Formally, the reachability specification for a layer $l\in[0,\lmax-1]$ in context $\nu\in\Y{l+1}$ w.r.t.\ the next context $\nu'\in\ON{Post}^{l+1}(\nu)$ is defined by 
 \begin{align}\label{equ:psi}
  \psi^{l}_{\nu}(\nu')&:=
  \DiCases{
  \SetComp{w\in\INTERSECT{\Tuple{\X{l}_\nu\times\Y{l}_{\nu}}^*}{\Play{l}_{\nu}}}{\maxw{w}\in\INTERSECT{\Yla{l}{\nu}}{\Yli{l}{\nu'}}}}
  {\nu\neq \nu'}
  {\Set{\INTERSECT{\Tuple{\X{l}_\nu\times\Yli{l}{\nu}}^{\omega}}{\Play{l}_{\nu}}}}
  {\nu= \nu'}
 \end{align}
and the combination of $\psi^{l}_{\nu}(\nu')$ with a local task $\varphi^l_\nu\in\PhiallR$ is defined by
\begin{equation}\label{equ:ConcatSpec}
 \phisconc{l}{\nu}{\nu'}
 :=\SetCompX{\xi\sconc\xi'}{
  \propConj{
  \xi\in\varphi^l_\nu}{
  \maxw{\xi}\sconc\xi'\in\psi^l_\nu(\nu')}
 }.
\end{equation}

\begin{example}
 Consider the floor plan in \REFfig{fig:contexts} and assume that the robot is in state $q^5_{22}$ corresponding to the states $r^5_{11}$ and $f^5$ in layers $l=1$ and $l=2$, respectively, as indicated by the light gray coloring. Now assume that the controller in layer $l=2$ requests a context change from $f^5$ to $f^6$. This induces the reachability specification $\psi^{1}_{f^5}(f^6)$ containing all sequences of rooms in $\Play{1}_{f^5}$ with final room $s_{56}$.
 Now a memoryless strategy for this specification first needs to request a context change from $r^5_{11}$ to $r^5_{21}$. This request, in turn, induces the reachability specification $\psi^{1}_{r^5_{11}}(r^5_{21})$ containing all sequences of cells in $\Play{}_{r^5_{11}}$ with final cell $q^5_{43}$. A possible first move of the robot to fulfill this specification is from $q^5_{22}$ to $q^5_{32}$. The respective goal states of the two specifications are indicated in dark gray in \REFfig{fig:contexts}.
\end{example}

%

The construction in \eqref{equ:ConcatSpec} implies that only a (possibly strict) prefix $\xi$ of a play $\pi\in\phisconc{l}{\nu}{\nu'}$ needs to be contained in $\varphi^l_\nu$. While this might seem restrictive for non-suffix closed specifications such as safety, one can circumvent this problem by using the idea of \enquote{weak until}. Intuitively, one would specify to stay safe, i.e., only visit states from a set $Q_{\mathsf{safe}}$, \enquote{until} the context is left. Then \eqref{equ:ConcatSpec} checks if the current requested context change can be enforced by staying in safe states. For reachability type specifications, such as the request of the completion of a certain task, this issue does not arise. 

Given the above definitions of local specifications, hierarchical reactive games can be constructed from a set of LGGs as follows.


\begin{definition}\label{def:HierarchicalGames}
  Given a set of local specifications $\PhiallR$ over a set of LGGs $[\Glall]$ and a set of level $0$ initial states $\I{}\subseteq(\X{}\times\Y{})$, the tuple  $\Tuple{[\Glall],\I{},\PhiallR}$ is called a \emph{hierarchical reactive game} (HRG) over $[\Glall]$. 
  Furthermore, given the set of local initial conditions
  \begin{equation}\label{equ:Il}
  \I{l}(m):=\TriCases
  {\SetComp{\Tuple{\Rxup{l}(x,y),\Ryup{l}(y)}}{\Tuple{x,y}=\I{}}}{m=0}
  {\Set{\maxw{\plPdm{l}(m-1)}}}{m>0,l<L}
  {\text{undefined}}{else,}
 \end{equation}
 a set $\rhoall$ is defined to be \emph{winning} (resp. \emph{possibly winning}) for $\Tuple{[\Glall],\I{},\PhiallR}$, if for all $l\in[0,L-1]$ holds that
  \begin{compactenum}[(i)]
  \item for all $m\in\dom{\plP{l+1}}$ (with $m<\maxk{\plP{l+1}}$ if $\maxk{\plP{l+1}}<\infty$) there exists a prefix $\xi\sqsubseteq\plPdm{l}(m)$ s.t.\ $\xi$ is winning for $\Tuple{\Play{l}_{\plPy{l+1}(m)},\I{l}(m),\varphi^l_{\plPy{l+1}(m)}}$, and 
  \item for $m=\maxk{\plP{l+1}}<\infty$ there exists a string $\xi=\plPdm{l}(m)$ (resp. $\xi\sqsubseteq\plPdm{l}(m)$ or $\plPdm{l}(m)\sqsubseteq\xi$) s.t.\ $\xi$ is winning for $\Tuple{\Play{l}_{\plPy{l+1}(m)},\I{l}(m),\varphi^l_{\plPy{l+1}(m)}}$, and 
%
\item $\plP{\lmax}$ is winning (resp. possibly winning) for $\Tuple{\Play{\lmax},\I{L}(0),\varphi^\lmax}$.
 \end{compactenum}
 \vspace{-0.4cm}
\end{definition}

\section{Assume-Admissible Hierarchical Strategy Construction}\label{sec:Strategies}

Let $\Tuple{[\Glall],\I{},\PhiallR}$ be a HRG with initial condition $\I{}\in(\X{}\times\Y{})$ and let $\Alphaall$ be a set of local environment assumptions over $[\Glall]$. 
Then we want to synthesize a strategy (i.e., a controller) for layer $0$ that generates a play whose 
projection is winning for the set of local system specifications $\PhiallR$ if $\Alphaall$ holds.
We assume that $\PhiallR$ and $\Alphaall$ are both $\omega$-regular languages.
While in principle one can flatten the game and solve one global game to obtain a solution to this problem, this will be prohibitively expensive. 
We therefore propose an algorithm that 
constructs a winning strategy in each local game that is encountered and ``stitches together'' these winning strategies dynamically.
Additionally, one could statically solve and memorize all possibly constructed local games. 
Our algorithm avoids this expensive construction by only solving games that actually arise online.
Hence, our procedure is \emph{dynamic} in that it solves a series of local games in each step starting
from the current state --- this is conceptually similar to receding horizon control approaches.
To incorporate environment assumptions, 
we use a slightly modified version of the algorithm from \cite{BrenguierRaskinSankur_ArXiv_2015}
to compute an assume-admissible winning strategy for a local game and a local environment assumption.
Our procedure treats this algorithm as a black box; in principle, a different strategy synthesis algorithm can be used.

\subsection{Synthesis of Assume-Admissibly Winning Strategies}\label{sec:AAwinningSol}

Assume-admissibly winning strategies for the play $\Tuple{\G,\I{},\varphi}$ w.r.t.\ the assumption $\eass$ can 
be computed by the algorithm given by \citet[Thm. 4]{BrenguierRaskinSankur_ArXiv_2015} in case $\varphi$ and $\eass$
are $\omega$-regular objectives. 
We denote the outcome of this strategy synthesis by $\Sol{\mathsf{AA}}{\G,\I{},\varphi,\eass}$.
Whenever the environment does not play admissible, the definition of assume-admissibly winning strategies does only restrict the behavior of the system to an admissible one. This does not give any guarantees w.r.t.\ $\varphi$ in case the environment does not play admissible. To circumvent this issue we slightly modify the outcome of the available strategy synthesis.
%
\begin{definition}\label{def:SolAA}
 Let $f^{\mathsf{AA}}=\Sol{\mathsf{AA}}{\G,\I{},\varphi,\eass}$ be an assume-admissibly winning strategy, then its associated \emph{possibly winning strategy} $f$, is defined for all $\play{}\in\Play{}$ s.t.
  \begin{equation}\label{equ:Sol}
  \SplitX{f(\play{}|_{[0,k]},x(k\plps1))=}{\DiCases
  {f^{\mathsf{AA}}(\play{}|_{[0,k]},x(k\plps1))}{\play{}|_{[0,k]}\sconc\Tuple{x(k\plps1),f^{\mathsf{AA}}(\play{}|_{[0,k]},x(k\plps1))}\in\overline{\varphi}}
  {\emptyset}{\text{else}.}}
 \end{equation}
 We define the set of all possibly winning strategies for the game $\Tuple{\G,\I{},\varphi}$ w.r.t. $\eass$ by $\Sol{}{\G,\I{},\varphi,\eass}$.
\end{definition}
A strategy $f=\Sol{}{\G,\I{},\varphi,\eass}$ blocks whenever the environment forces the play into a state from which the play cannot be won anymore. This implies that all finite plays $\pi$ compliant with $f$ are possibly winning, i.e. $\play{}\in\overline{\varphi}$, even if the environment does not play admissible. However, if it does, the compliant play is winning. This is formalized by the following proposition.	
\begin{proposition}\label{prop:WeaklyWinning}
 Given $f=\Sol{}{\G,\I{},\varphi,\eass}$, $\g{}\in\StrategyE(\G)$, it holds for all $\play{}\in\Play{}$ that 
 \begin{subequations}
   \begin{align}
    &\propImp{\g{}\in\AdmStrat\Tuple{\G,\I{},\eass}}{f\in\WinStrat(\G,\I{},\varphi,\g{})}, \label{equ:WeaklyWinningAdm}\\
&\text{and}\quad\propImp{
\begin{propConjA}
\play{}\in\Compl\Tuple{\f{},\I{}}\\
 |\play{}|<\infty
\end{propConjA}
}{\play{}\in\WinPlays\Tuple{\G,\I{},\overline{\varphi}}}.\label{equ:WeaklyWinningNonAdm}
 \end{align}
 \end{subequations}
\end{proposition}
\begin{proof}
 Let $f^{\mathsf{AA}}=\Sol{\mathsf{AA}}{\G,\I{},\varphi,\eass}$ and $f$ its associated possibly winning strategy. Using \eqref{equ:AdmStrat}, 	$\g{}\in\AdmStrat\Tuple{\G,\I{},\eass}$ implies $f^{\mathsf{AA}}\in\WinStrat(\G,\I{},\allowbreak\varphi,\g{})$. Using \eqref{equ:WinStrat}, this implies $\pi\in\WinPlays\Tuple{\G,\I{},\overline{\varphi}}$. Therefore, the second case in \eqref{equ:Sol} cannot occur and we obtain $f=f^{\mathsf{AA}}$, i.e., $\f{}\in\WinStrat(\G,\I{},\varphi,\g{})$.
 Observe that the left side of \eqref{equ:WeaklyWinningNonAdm} implies that the right side of \eqref{equ:newcompliant} holds for $\play{}$ and $\f{}$, hence $f(\play{}|_{[0,k-1]},x(k))\neq\emptyset$ for all $k\in\dom{\pi}$. Using \eqref{equ:Sol}, this implies $\play{}|_{[0,k]}\in\Compl\Tuple{f^{\mathsf{AA}},\I{}}$ and $\play{}|_{[0,k]}\in\overline{\varphi}$, hence, $\pi\in\WinPlays\Tuple{\G,\I{},\overline{\varphi}}$.
%
%
%
\end{proof}

We remark that the algorithm to compute assume-admissible strategies  in \citet[Thm. 4]{BrenguierRaskinSankur_ArXiv_2015} 
can be trivially adapted to ensure \REFprop{prop:WeaklyWinning}, by blocking the game whenever a losing state 
(one in which there is no winning strategy for the system)
is entered.


\subsection{The Strategy Synthesis Algorithm}\label{sec:Strategy:Construct}


Recall that we aim to synthesize a strategy (i.e., a controller) for layer $0$ that generates a play whose projection 
is assume-admissible winning for the HRG $\Tuple{[\Glall],\I{},\PhiallR}$ w.r.t.\ $\Alphaall$. 
Hence, the goal of each computation round of our algorithm is to determine the next system state $y(k+1)$ in layer $0$, 
i.e., to calculate the current control action that needs to be applied to the system. 
This depends on the environment state $x(k+1)$ in layer $0$ which is sensed in the beginning of each such computation round 
and projected to all layers $l\in[1,L]$ in an \enquote{bottom up} fashion. 
The current state in every layer local game is given by the restriction of $x^l(k+1)$ to the current context and the projection $y^l(k)$ of the last system state.
Based on this information, the next step in every layer local game needs to be calculated. 

This calculation is challenging due to the interaction between plays in different layers. 
In particular, a move from system state $\nu$ to $\nu'$ requested by a strategy in layer $l\in[1,L]$ results in an additional reachability 
specification for the current local game in layer $l-1$. 
Furthermore, such an \enquote{induced} reachability specification for the local game in layer $l-1$ and context $\nu$ 
might change multiple times, before this context is left.
This is due to the fact that an environment state in layer $l>0$ possibly changes multiple times before a system state change follows, as 
discussed in the construction of abstract game graphs (see \REFsec{subsec:AGG}). Hence, whenever such a specification change occurs, the strategy in layer $l-1$ needs to be re-calculated. The only strategy that is not influenced by this interplay is the highest level strategy, which is computed only once when initializing the algorithm. 
Once the strategies are updated in a \enquote{top down} manner, the controller picks the next move at layer $0$ based on the updated strategy for layer $0$ and plays it. This changes the states for all higher layers and the algorithm continues with the next computation cycle.

We now describe the algorithm formally.

\begin{algo}[Strategy Synthesis Procedure]\label{def:F}
Let $\Tuple{[\Glall],\I{},\PhiallR}$ be a HRG with $\I{}\in(\X{}\times\Y{})$ and $\Alphaall$ a set of local environment assumptions over $[\Glall]$.
Then the dynamic hierarchical strategy $F=\Set{f^l}_{l=0}^L$ for the game $\Tuple{[\Glall],\I{},\PhiallR}$ w.r.t.\ $\Alphaall$ and its compliant play $\pi$ are iteratively defined as follows:\\
\begin{compactitem}[$\blacktriangleright$]
 \item Initialization:\\
 \begin{subequations}\label{equ:F}
 \begin{inparaitem}[$\triangleright$]
  \item Using $\I{L}$ as in \eqref{equ:Il}, calculate the assume admissible winning strategy for the highest layer $L$ using
  \begin{equation}
   h^L=\Sol{}{\G^L,\I{L}(0),\varphi^L,\eass^L}.
  \end{equation}
\item Initialize the play and the local history, respectively, with 
\begin{equation}
 \pi=\Tuple{x(0),y(0)}=\I{}\quad\text{and}\quad\plg{l}(0)=\pi.
\end{equation}
 \end{inparaitem}
\item Iteration for all $k\in\Nb$:\\
\begin{inparaitem}[$\triangleright$]
 \item Sense the environment move 
 \begin{equation}
  x(k+1)\in\TrE{0}{}(\pi).\label{equ:F:xk}
 \end{equation}
 \item Compute the local environment state $\plxd{l}(k+1)$ using \eqref{equ:pil} and \eqref{equ:pld}, i.e.,
 \begin{equation}\label{equ:F:plxd}
  \plxd{l}(k+1)=\rx{l}{y^{l+1}(k)}\Tuple{\Rxup{l}(x(k+1),y(k))}
 \end{equation}
 for each layer $l$;\\
 \item Iteratively calculate the current strategy by 
 \begin{align}
  f^L(k)&=h^L\quad\text{and}\label{equ:F:fLk}\\
 \AllQ{l\in[0,L-1]}{f^{l}(k)&=\TriCases 
 {\emptyset}{
 \GotStuck{l+1}(k)
 }
 {h^l(k)}{
 \Done{l+1}(k)
 }
 {f^l_{\nu\nu'^{l+1}}(k)}{\text{else}}
 }\label{equ:F:flk}
\end{align}
with 
\begin{align}
\nu&:=y^{l+1}(k),\notag\\
\nu'^{l+1}(k)&:=f^{l+1}(k)(\plg{l+1}(k),\plxd{l+1}(k\plps1)),\label{equ:F:nu}\\
 h^l(k)&:=\DiCases
 {\Sol{}{\G^l_{\nu},\Set{\plg{l}(k)},\varphi^l_{\nu},\eass^l_{\nu}}}
 {\nu\neq y^{l+1}(k-1)}
 {h^l(k-1)}{\text{else}}\label{equ:F:hlk:DC}\\
 f^l_{\nu\nu'^{l+1}}(k)&=
  \DiCases
  {\Sol{}{\G^l_{\nu},\Set{\plg{l}(k)},\phisconc{l}{\nu}{\nu'^{l+1}(k)},\eass^l_{\nu}}}
  {\begin{propDisjA}
    \nu\neq y^{l+1}(k\mips1)\\
    \nu'^{l+1}(k)\neq\nu'^{l+1}(k\mips1)
   \end{propDisjA}}
  {f^l_{\nu\nu'^{l+1}}(k-1)}{\text{else}}\label{equ:F:flk:DC}
\end{align}
 and the predicates are defined by
 \begin{align}
  &\propAequ{\Win{l}(k)}{
  \plg{l}(k)\in\DiCases{\varphi^{l}_{\nu}}{l\in[0,L-1]}{\varphi^L}{l=L}
},\label{equ:F:Win}\\
  &\propAequ{\Done{l}(k)}{\begin{propConjA}
  \propDisj*{l=L}{\Done{l+1}(k)}\\
  \Win{l}(k)\\
(\plg{l}(k),\plxd{l}(k+1))\notin \dom{h^l(k)}
 \end{propConjA}
},\quad\text{and}\label{equ:F:Done}\\
  &\propAequ{\GotStuck{l}(k)}{
\begin{propConjA}
\neg\Done{l}(k)\\
(\plg{l}(k),\plxd{l}(k+1))\notin\dom{f^{l}(k)}
\end{propConjA}
  }.\label{equ:F:GotStuck}
\end{align}

 \item Play the next move following the current system strategy for layer $l=0$
 \begin{equation}
  y(k+1)=f^0(k)(\plg{0}(k),\plxd{0}(k+1)).\label{equ:F:yk}
 \end{equation}
 \item Append $\Tuple{x(k+1),y(k+1)}$ to the play giving 
 \begin{equation}
  \pi=\Tuple{x\ll{0,k+1},y\ll{0,k+1}}\label{equ:F:pi}.
 \end{equation} 
 \item Using \eqref{equ:plPdm:b}, compute the new context restricted history 
 \begin{equation}
  \plg{l}(k+1)=\maxw{\plPdm{l}}\quad\text{with}\quad\plPdm{l}\in\rhoall.\label{equ:F:LH}
 \end{equation}
\end{inparaitem}
\end{subequations}
\end{compactitem}
\end{algo}
%
%

As discussed before, every computation round $k$ of the construction in \eqref{equ:F} starts with the sensing of 
the next environment move in \eqref{equ:F:xk}, giving the full $0$-level environment state $x(k+1)=x^0(k+1)$. 
This state is used to compute the local restricted environment states $\plxd{l}(k+1)$ for every layer and 
current context $y^{l+1}(k)$ in \eqref{equ:F:plxd}. 
Note that this construction is done \enquote{bottom up}.

Thereafter, the selection of the current strategy $f^l$ for every layer and its respective current goal state $\nu'^l$ are calculated. 
Observe that this is done \enquote{top down}, as $\nu'^l$ is used to calculated the current reachability specification for the reachability 
game in layer $l-1$. 
The construction of $f^l$ in \eqref{equ:F:flk} distinguishes three cases: the play at the highest layer has been won,
or the play at the higher layer got stuck,
or none of these conditions occurred. We consider these cases separately.

For the first case observe, that the specification of level $L$ might be a set of finite strings and local specifications are sets of 
finite strings by definition (see \REFsec{subsec:DefHierarchicalGames}). Therefore, the play 
constructed in \eqref{equ:F} does not need to be infinite to be winning for $\PhiallR$. 
If the play in layer $L$ is winning for $\varphi^L$ and the strategy does not request any 
other move (denoted by the predicate $\Done{L}$ in \eqref{equ:F:Done}), then this is communicated 
downwards using the second line of \eqref{equ:F:flk}. 
In this case all lower level strategies must be winning for local specifications only, 
using the assume-admissible strategy calculated in \eqref{equ:F:hlk:DC}.

For the second case, observe that the strategy calculation in \eqref{equ:F:hlk:DC} and \eqref{equ:F:flk:DC} does not need to have a solution.
Further, even if it has a solution, system strategies are not assumed to be left-total. 
Hence, there might exist (non-admissible) environment moves that cause a blocking of $f$ without the game being winning. 
These two situations are modeled by the predicate $\GotStuck{l}$ in \eqref{equ:F:Done}. 
If such a situation occurs, it is communicated downwards by the first line of \eqref{equ:F:flk} 
resulting in $\GotStuck{l'}$ for all $l'<l$ and therefore an abortion of the game. 
Intuitively, the first time $\GotStuck{l}$ occurs, it is because of 
an \enquote{unrealizeable} local specification.
We introduce a fourth predicate 
\begin{equation}\label{equ:F:UnReal}
 \propAequ{\UnReal{l}(k)}{
 \DiCases{\GotStuck{l}(k)}{l=L}
 {\propConj{\neg\GotStuck{l+1}(k)}{\GotStuck{l}(k)}}{l<L}}
\end{equation}
to remember the first layer at which the controller got stuck.
We will show in \REFsec{sec:Strategy:Soundness} that an unrealizable specification is the only reason for a non-winning play constructed in \eqref{equ:F} to be aborted. 

In the third case, i.e., if neither $\GotStuck{l}$ nor $\Done{l+1}$ is true, the strategy for level $l$  is calculated by \eqref{equ:F:flk:DC} 
using again two subcases. 
In the first subcase, either a new context was entered (resulting in a new local game) or the \enquote{top down induced} reachability specification has 
changed (due to a change of $\nu'^l$ caused by a new environment state in layer $l+1$). In this case the strategy for level $l$ needs to be re-calculated. 
However, if neither of these two situations occurs, the strategy from the previous time step can be used, 
avoiding unnecessary re-computations.
 
After the strategy construction in \eqref{equ:F:flk}-\eqref{equ:F:GotStuck}, the system state is updated to $y(k+1)$, 
using the currently selected lowest level strategy $f^0(k)$ in \eqref{equ:F:yk}. 
Hence, \eqref{equ:F:flk}-\eqref{equ:F:GotStuck} only utilize the hierarchical structure of the game graph to 
compute $f^0(k)$, which is the only control action that is actually applied to 
the system, e.g., the robot in our example. Then $\Tuple{x(k+1),y(k+1)}$ is appended to the constructed play $\pi$. 
As intuitively assumed, such plays $\pi$ generated by \REFalg{def:F} up to length $k$ are plays in $\G$, i.e., $\pi\in\Play{}$, as shown in the following proposition. Observe, that this implies that also $\plP{l}\in\Play{l}$ for all $l\in[0,L]$ (from \REFprop{prop:Play-l}) and $\plPdm{l}(m)\in\Play{l}_{\plPy{l+1}(m)}$ for all $l\in[0,L-1]$ and $m\in\dom{\plP{l+1}}$ (from \REFprop{prop:playGly_simple}).

\begin{proposition}\label{prop:piInPlay}
Let $\pi$ be a play computed in \REFalg{def:F}. Then $\pi\in\Play{}$.
\end{proposition}

\begin{proof}
It follows from \eqref{equ:F:xk} and \eqref{equ:F:yk} that
     \begin{equation}\label{equ:proof:l0}
      \AllQ{k\in\domp{\pi}}{
      \begin{propConjA}
       x(k)\in\TrE{}{}(x(k-1),y(k-1))\\
       y(k)=f^{0}(k-1)(\plg{0}(k-1),\plxd{0}(k))
      \end{propConjA}},
     \end{equation}
implying $f^0(k-1)\neq\emptyset$ for all $k\in\domp{\pi}$. Therefore, \eqref{equ:F:flk}-\eqref{equ:F:GotStuck} imply that $f^0(k-1)$ is a system strategy over $\mathcal{G}^0_{y^{1}(k-1)}$ and the definition of the latter in \REFsec{sec:Prelim} gives
  $f^0(k-1)(\plg{0}(k-1),\plxd{0}(k))\in\TrS{0}{y^{1}(k-1)}(\plxd{0}(k-1),\maxw{\plg{0}(k-1)}_2)$.
 Now observe from \eqref{equ:F:LH}, \eqref{equ:plPdm:b} and \eqref{equ:projpi} that $\maxw{\plg{0}(k-1)}_2=y^0(k-1)$. 
 Now using $\TrS{0}{y^{1}(k-1)}\subseteq\TrS{0}{}$ from \REFass{ass:Xbinlnu} along with this observation, we see that \eqref{equ:proof:l0} actually implies \eqref{equ:playp_def:b}, hence $\pi\in\Play{}$. 
\end{proof}
%
We call a play 
$\pi$ calculated in \eqref{equ:F} up to length $k=|\play{}|$ \emph{maximal} if 
\begin{equation}
 \propImp{k<\infty}{(\plg{0}(k),\plxd{0}(k+1))\notin\dom{f^{0}(k)}}.
\end{equation}

One round of the construction in \eqref{equ:F} is ended by calculating the current local histories $\plg{l}(k+1)$ for every layer.
Intuitively, $\plg{l}(k+1)$ models the part of $\plP{l}$ generated after the last context change in layer $l$ and is therefore equivalent to $\maxw{\plPdm{l}}$. These histories are used in the calculation of assume-admissible strategies to ensure that a re-computation of a strategy within one context does result in a continuation of the already generated string w.r.t.\ the given specification.

While the local system strategies $f^l(k)$ are explicitly calculated for every time step $k$ in \eqref{equ:F:flk}-\eqref{equ:F:GotStuck}, the local environment strategies $g^l(k)$ are only given implicitly by the observed environment move \eqref{equ:F:xk} and its abstraction to every layer $l$. Formally, a play $\pi$ calculated in \eqref{equ:F} was played against an admissible environment strategy if for all $l\in[0,L-1]$, $m\in\dom{\plP{l}}$ exists an environment strategy $g^l_{\plPy{l+1}(m)}\in\AdmStrat(\G^l_{\plPy{l+1}(m)},\I{l}(m),\eass^l_{\plPy{l+1}(m)})$ s.t. $\plPdm{l}(m)\in\Compl(\G^l_{\plPy{l+1}(m)},g^l_{\plPy{l+1}(m)})$ and for layer $L$ exists $g^L\in\AdmStrat(\allowbreak\G^L,\I{L}(0),\eass^L)$ s.t. $\plP{L}\in\Compl(\G^L,g^L)$. If this holds, we call $\pi$ an \emph{environment admissible play}. 



\begin{example}
 Consider the play $\pi$ whose $y$-component is depicted by filled cycles in \REFfig{fig:layers} (bottom) and (for simplicity) the static environment used in \REFexp{exp:plPdm}, where we use $o=\Set{q^5_{24},q^5_{25},q^5_{63}}$ and $o_{\downarrow}=\Set{q^5_{24},q^5_{25}}$ for notational convenience. In this game the only objective is to reach $q^6_{63}$ in $r^6_{21}$ and $f^6$. This implies that $\PhiallR$ contains only empty sets except for 
 \begin{align*}
  \varphi^2=\Set{\bot}\times\Set{f^5f^6},~\varphi^1_{f^6}=\Set{\bot}\times R^*\sconc\Set{r^6_{21}},~\text{and}~\varphi^0_{r^6_{21}}=\Set{\bot}\times Q^*\sconc\Set{q^6_{21}}. 
 \end{align*}
 To illustrate \REFalg{def:F} we pick $k=2$, i.e., $\pi$ was generated for $3$ time steps and we are now calculating $\pi(3)=\Tuple{x(3),y(3)}$ using \eqref{equ:F}. \\
 First recall from \REFexp{exp:plPdm} that
 \begin{align*}
  \pi(2)&=\Tuple{o,q^5_{33}},~\pi^1(2)=\pi^1(0)=\Tuple{\Set{d},r^5_{11}},~\pi^2(2)=\pi^2(0)=\Tuple{\Set{\bot},f^5},~\text{and}\\
  \plg{0}(2)&=\Tuple{o_{\downarrow},q^5_{22}}\Tuple{o_{\downarrow},q^5_{23}}\Tuple{o_{\downarrow},q^5_{33}},~ 
  \plg{1}(2)=\Tuple{\Set{d},r^5_{11}},~\plg{2}(2)=\Tuple{\Set{\bot},f^{5}}.
 \end{align*}
 We furthermore assume that the strategy calculation for $k=0$ resulted in the requested moves depicted by the arrows in \REFfig{fig:contexts} (middle and top).
  Whith this initialization we obtain the following steps of the algorithm.\\
  \begin{inparaitem}[$\triangleright$]
   \item Due to the static environment assumption, \eqref{equ:F:xk} gives $x(k+1)=x(3)=o$.\\
   \item Applying \eqref{equ:F:plxd} yields $\plxd{0}(3)=o_{\downarrow}$, $\plxd{1}(3)=\Set{d}$ and $\plxd{2}(3)=\Set{\bot}$.\\
   \item First, \eqref{equ:F:fLk} and \eqref{equ:F:fLk} imply $f^2(2)\neq\emptyset$, $\nu'^{2}(2)=\nu'^{2}(1)=f^6$ and $\neg\Done{2}(2)$. Therefore, \eqref{equ:F:flk:DC} and \eqref{equ:F:flk} imply $f^1(2)=f^1_{f^5f^6}(0)\neq\emptyset$, $\nu'^{1}(2)=\nu'^{1}(1)=r^5_{11}$ and $\neg\Done{1}(2)$. With this, the lowest level strategy is given by $f^0(2)=f^0_{r^5_{11},r^5_{21}}(0)$.\\
   \item As we assume a static environment and no obstacles block the way between the robot and the exit to room $r^5_{21}$, we assume that $f^0_{r^5_{11},r^5_{21}}$ is a shortest path strategy and \eqref{equ:F:yk} gives $y(k+1)=y(3)=q^5_{43}$.\\
   \item Observe, that a context change has occurred during this step, i.e., \eqref{equ:F:LH} gives
   \begin{align*}
    \plg{0}(3)&=\Tuple{\plxd{0}(3),y(3)}=\Tuple{o_{\downarrow},q^5_{43}},~\plg{1}(3)=\Tuple{\Set{d},r^5_{11}}\Tuple{\Set{d},r^5_{21}},~
    \plg{2}(3)=\Tuple{\Set{\bot},f^{5}}.
   \end{align*}
  \end{inparaitem}
  With this local history the next iteration of the algorithm is started.  For the assumed very simple static environment, \REFalg{def:F} will never get stuck. Observe, that once we reach floor $f^6$, the level $2$ game is won and $\Done{2}$ is true. In this case $h^1$ will be calculated w.r.t.\ the specification $\varphi^1_{f^6}$. If in addition $r^6_{21}$ is reached, $\Done{1}$ is also set to true and $h^0$ is calculated. After one more time step also $\Done{0}$ is true and the algorithm terminates. The generated play is obviously winning for $\PhiallR$.
\end{example}

\subsection{Soundness}\label{sec:Strategy:Soundness}
In this section we prove three different soundness results for the play constructed in \REFalg{def:F}. Intuitively, \REFalg{def:F} is sound if a play $\pi$ calculated in \eqref{equ:F} is winning for the HRG $\Tuple{[\Glall],\I{},\PhiallR}$ if all generated local specifications are realizable and the environment plays admissible w.r.t. $\Alphaall$, which will be proven last in \REFthm{thm:FiniteWinningPlays}. 
As a first intermediate result we show that the only two reasons for a maximal play to terminate are actually that 
\begin{inparaenum}[(i)]
 \item a current local specification is not realizable or
 \item the play is already winning given a finite winning condition in layer $L$.
\end{inparaenum}

\begin{theorem}\label{thm:GotStuck}
Let 
$\pi$ be a \emph{maximal}
play computed by \eqref{equ:F}. Then it holds that
\begin{align}\label{equ:thm:GotStuck}	
 \propAequ{|\pi|<\infty}{
  \begin{propDisjA}
\AllQ{l\in[0,L]}{\Done{l}(\maxk{\pi})}\\ 
\ExQ{l\in[0,L]}{\UnReal{l}(\maxk{\pi})}\\
 \end{propDisjA}
 }.
\end{align}
\end{theorem}

\begin{proof}
To prove this theorem we need that 
\begin{subequations}\allowdisplaybreaks
\begin{equation}\label{equ:thm:GotStuck:1}
 \propAequ{
\ExQ*{l\in[0,L]}{\UnReal{l}(\maxk{\pi})}\\
 }
 {\GotStuck{0}(\maxk{\pi})}
\end{equation}
which is proven for all $k\in\dom{\pi}$ in \REFlem{lem:GotStuck:UnRealCanNotWait} (see \REFapp{sec:app:proofs}). Furthermore, as we assume environment strategies to be left-total, \eqref{equ:F:xk} can always be computed. Hence, $\pi$ becomes finite while being maximal iff \eqref{equ:F:yk} cannot be evaluated, i.e.,
\begin{equation}\label{equ:FiniteCausedByf} 
 \propAequ{\maxk{\pi}<\infty}{(\plg{0}(\maxk{\pi}),\plxd{0}(\maxk{\pi}+1))\notin\dom{f^{0}(\maxk{\pi})}}.
\end{equation}
Now we pick $k=\maxk{\pi}$ and prove both directions separately.\\
\begin{inparaitem}[$*$]
\item[\enquote{$\Rightarrow$}]
Using \eqref{equ:FiniteCausedByf} and \eqref{equ:F:GotStuck} implies that either 
\begin{inparaenum}[(i)]
 \item $\neg\Done{0}(k)$ and $\GotStuck{0}(k)$, or
 \item $\Done{0}(k)$.
\end{inparaenum}
Using \eqref{equ:thm:GotStuck:1}, (i) implies\footnote{To simplify notation we denote by $\langle$(\#).right.n$\rangle$ (resp. $\langle$(\#).left.n$\rangle$) the $n$th statement on the right (resp. left) side of the implication/equivalence relation in equation (\#).} \eqrefR{equ:thm:GotStuck}{2}.
As $\Done{0}(k)$ implies $\AllQ{l\in[0,L]}{\Done{l}}(k)$ (from \eqref{equ:F:Done}), (ii) implies \eqrefR{equ:thm:GotStuck}{1}. \\
\item[\enquote{$\Leftarrow$}]
If \eqrefR{equ:thm:GotStuck}{2} is true, it follows from \eqref{equ:thm:GotStuck:1} that $\GotStuck{0}(k)$ and $\neg\Done{0}(k)$ (see the proof of \REFlem{lem:GotStuck:UnRealCanNotWait}). Hence, \eqref{equ:thm:GotStuck} and \eqref{equ:FiniteCausedByf} implies \eqrefL{equ:thm:GotStuck}{}.
If \eqrefR{equ:thm:GotStuck}{1} is true, we know from \eqref{equ:F:flk} that $f^{0}(k)=h^0(k)$. Therefore, \eqrefR{equ:F:Done}{3} and \eqref{equ:FiniteCausedByf} implies \eqrefL{equ:thm:GotStuck}{}.
\end{inparaitem} 
\end{subequations}
\end{proof}

While the second case in \REFthm{thm:GotStuck} is not desired w.r.t.\ the goal of constructing a winning play, it can usually not be avoided in a realistic scenario as we can 
\begin{inparaenum}[(i)]
 \item not enforce the environment to play admissible and 
 \item checking feasibility of all possibly occurring local games before startup might not be appropriate, as this set might be very large.
\end{inparaenum}
However, \REFalg{def:F} ensures that if this situation occurs, the local specifications are not falsified up to this point. This is formalized by the notion of possibly winning, which ensures that generated finite plays always stay in the prefix closure of the considered local specifications. 

\begin{theorem}\label{thm:SoundnessRecursivePlay}
Given the preliminaries of \REFalg{def:F}, let $\play{}$ be the play computed by \eqref{equ:F} up to length $k$, 
and $\rhoall$ its set of local projected play sequences. Then $\rhoall$ is possibly winning for $\Tuple{[\Glall],\I{},\PhiallR}$.
\end{theorem}

\begin{proof}
We have two important observations that we use in this proof. First, it holds for all $l\in[0,\lmax]$ and $m\in\domp{\plP{l}}$ that 
\begin{subequations}
  \begin{align}\label{equ:proof:linit:D}
 \begin{propConjA}
 \plPxd{l}(m)\in\TrE{l}{\plPy{l}(m-1)}(\plPxd{l}(m-1),\plPy{l}(m-1))\\
 \plPy{l}(m)=f^{l}(\kappa^{l}(m)-1)(\plg{l}(\kappa^{l}(m)-1),\plPxd{l}(m))
 \end{propConjA}
\end{align}
as proven in \REFlem{lem:SoundnessRecursivePlay:A} (see \REFapp{sec:app:proofs}). Second, 
it holds for all $l\in[0,\lmax-1]$ and $m\in\domp{\plP{l+1}}$ that
  \begin{align}
  \plPdm{l}(m-1)&\in\phisconc{l}{\plPy{l+1}(m-1)}{\plPy{l+1}(m)}\label{equ:localWin:a:D}
  \end{align}
and for $m=\maxk{\plP{l+1}}$ there exists $\nu'\in\ON{Post}^{l+1}(\plPy{l+1}(m))$ s.t.
\begin{align}
   \plPdm{l}(m)&\in\overline{\phisconc{l}{\plPy{l+1}(m)}{\nu'}},\label{equ:localWin:b:D}
\end{align}
\end{subequations}
as proven in \REFlem{lem:localWin} (see \REFapp{sec:app:proofs}).\\
Recall from \REFprop{prop:piInPlay} that $\pi\in\Play{}$, hence \REFprop{prop:playGly_simple} implies $\plPdm{l}(m)\in\Play{l}_{\plPy{l+1}(m)}$ and \eqref{equ:plPdm:a} obviously gives $\plPdm{l}(m)\ll{0,0}=\maxw{\plPdm{l}(m-1)}=\I{l}(m)$ for all $m\in\domp{\plP{l+1}}$. 
As \eqref{equ:localWin:a:D} holds, \eqref{equ:ConcatSpec} implies 
\begin{subequations}
  \begin{equation}\label{proof:thm:SoundnessRecursivePlay:1}
   \ExQ{\xi\in\overline{\Set{\plPdm{l}(m-1)}}}{\xi\in\varphi^{l}_{\plPy{l+1}(m-1)}}.
 \end{equation}
Now consider $m=\maxk{\plP{l+1}}$. As \eqref{equ:localWin:b:D} holds, \eqref{equ:ConcatSpec} implies that either
\begin{align}
 \plPdm{l}(m)\in\overline{\varphi^l_{\plPy{l+1}(m)}}\quad\text{or}\quad
 \ExQ{\xi\in\overline{\Set{\plPdm{l}(m)}}}{\xi\in\varphi^{l}_{\plPy{l+1}(m)}}.\label{proof:thm:SoundnessRecursivePlay:3}
\end{align}
\end{subequations}
Using the definitions of winning from \REFsec{sec:Prelim}, \eqref{proof:thm:SoundnessRecursivePlay:1}-\eqref{proof:thm:SoundnessRecursivePlay:3} imply that
conditions (i)-(ii) for possibly winning HRGs from \REFsec{subsec:DefHierarchicalGames} hold.
To prove condition (iii), observe from \eqref{equ:F:fLk} that $\AllQ{k\in\Nb}{f^L(k)=h^L}$. Furthermore, recall from the definition of $\rhoall$ that $\plPdm{L}(0)=\plP{L}$ and $\maxk{\plPdm{L}}=0$ and therefore $\plg{L}(\kappa^{l}(m)-1)=\plP{L}\ll{0,\kappa^{l}(m)-1}$. Using these observations in \eqref{equ:proof:linit:D}, it follows that \eqref{equ:newcompliant} holds for $\plP{L}$ w.r.t. $h^L$ and $\I{L}(0)$, implying $\play{L}\in\overline{\Compl\Tuple{h^L,\I{L}(0)}}$. As $h^L=\Sol{}{\G^L,\I{L}(0),\varphi^L,\eass^L}$ and $\plP{L}\in\Play{L}$ (from \REFprop{prop:piInPlay} and \REFprop{prop:Play-l}), it follows from \eqref{equ:WeaklyWinningNonAdm} in \REFprop{prop:WeaklyWinning} that $\plP{L}$ is possibly winning for $(\Play{L},\I{L}(0),\varphi^L)$.
\end{proof}

We now prove the main result of this paper, namely that maximal plays $\pi$ calculated by \REFalg{def:F} (finite and infinite) are actually winning for $\Tuple{[\Glall],\I{},\PhiallR}$ if the environment plays admissible and all constructed local plays have a solution, i.e.,
\begin{equation}\label{equ:ass2}
 \AllQ{k\in\dom{\pi},l\in[0,L]}{\neg\UnReal{l}(k)}.
\end{equation}

%
%

%

\begin{theorem}\label{thm:FiniteWinningPlays}
Let $\pi$ be a \emph{maximal and environment admissible} play computed by \eqref{equ:F} s.t. \eqref{equ:ass2} holds and let
$\rhoall$ be its set of local play sequences. Then $\rhoall$ is winning for $\Tuple{[\Glall],\I{},\PhiallR}$.
\end{theorem}
\begin{proof}
In this proof we use the following two observations
\begin{subequations}\label{equ:proof:FWP}\allowdisplaybreaks
 \begin{align}
&\AllQ*{k\inps\dom{\pi},l\inps[0,L]}{\neg\Done{l}(k)}
\Leftrightarrow\BR{|\pi|\eqps\infty}
\Leftrightarrow\AllQ*{l\inps[0,L]}{|\plP{l}|\eqps\infty},\label{equ:proof:FWP:a}\\
&\AllQ*{l\inps[0,L]}{\Done{l}(\maxk{\pi})}
\Leftrightarrow\BR{|\pi|<\infty}
\Leftrightarrow\AllQ*{l\inps[0,L]}{|\plP{l}|<\infty}.\label{equ:proof:FWP:b}
\end{align}
\end{subequations}
where \eqref{equ:proof:FWP:a} was proven in \REFlem{lem:GotStuck:Exist:B} (see \REFapp{sec:app:proofs}), the left side of \eqref{equ:proof:FWP:b} follows from \REFthm{thm:GotStuck} and \eqref{equ:ass2}, and the right side of \eqref{equ:proof:FWP:b} is a simple consequence from the definition of projections in \eqref{equ:projpi}. Hence, we generally have two cases to consider when proving the three conditions for winning HRGs from \REFsec{subsec:DefHierarchicalGames}.\\
First observe that condition (i) is equivalent for winning and possibly winning, no matter whether $\pi$ is finite or not. It therefore follows directly from \REFthm{thm:SoundnessRecursivePlay}. Furthermore, condition (ii) only needs to be proven if $|\plP{l+1}|<\infty$ and recall that for this case 
\REFthm{thm:SoundnessRecursivePlay} shows that $\plPdm{l}(\maxk{\plP{l+1}})$ is \emph{possibly} winning for $(\Play{l}_{\plPy{l+1}(m)},\maxw{\plPdm{l}(m-1)},\allowbreak\varphi^l_{\plPy{l+1}(m)})$ for all $l\in[0,L]$. Now observe from \eqref{equ:proof:FWP:b} that $\Done{l}(\maxk{\pi})$ which implies from \eqref{equ:F:Done} and \eqref{equ:F:Win} that $\plPdm{l}(\maxk{\plP{l+1}})=\plg{l}(\maxk{\pi})\in\varphi^{l}_{\plPy{l+1}(m)}$, where the first equality follows from \eqref{equ:F:LH} and \eqref{def:plPdm}. This obviously implies that \linebreak$\plPdm{l}(\maxk{\plP{l+1}})$ is winning in the above game. For finite plays, this reasoning also proves condition (iii). We therefore assume $|\plP{L}|=\infty$ and recall from the proof of \REFthm{thm:SoundnessRecursivePlay} that \eqref{equ:newcompliant} holds for $\plP{L}$ w.r.t. $h^L$ and $\I{L}(0)$. As $|\plP{L}|=\infty$ we have $\plP{L}\in\Compl\Tuple{h^L,\allowbreak\I{L}(0)}$. As $h^L=\Sol{}{\G^L,\I{L}(0),\varphi^L,\eass^L}$ and $\plP{L}\in\Play{L}$ (from \REFprop{prop:piInPlay} and \REFprop{prop:Play-l}) and $g^L\in\AdmStrat\Tuple{\G^L,\I{L}(0),\varphi^L,\eass^L}$, it follows from \eqref{equ:WeaklyWinningNonAdm} in \REFprop{prop:WeaklyWinning} that $\plP{L}$ is winning for $(\Play{L},\I{L}(0),\varphi^L)$.
\end{proof}

The important difference between \REFthm{thm:SoundnessRecursivePlay} and \REFthm{thm:FiniteWinningPlays} is that environment admissible infinite plays can only be generated if layer $L$ does not win in finite time, i.e., $\neg\Done{L}(k)$ for all $k\in\dom{\plP{L}}$. If the environment does not play admissible, infinite plays can also be generated if $\Done{L}(k)$ is true, as the environment might never \enquote{help} to reach the specification (i.e., does not play admissible) but also never moves to a losing state (i.e., causing the game to be aborted).

\begin{remark}
 It should be noted that the algorithm in \REFalg{def:F} works identically if we use a \enquote{usual} synthesis techniques to calculate 
winning (instead of assume-admissibly winning) strategies in $\Sol{}{\cdot}$ (i.e., a procedure to solve the unconstrained synthesis problem). 
Such a procedure is obtained, e.g., from the methods by \cite{Zielonka,EJ91} for general $\omega$-regular conditions, 
or more specialized procedures for co-safe properties (given by sets of finite-length plays) by
\cite{KV01,Finkbeiner,KW12}. 
This outlines the modularity of our approach w.r.t.\ the actual strategy synthesis routine used in local games.
However, it should be noted that in realistic scenarios, local games will usually not 
have winning strategies against a purely adversarial environment. 
Nevertheless, if the game gets stuck due to such an unrealizable sub-game, 
the result from \REFthm{thm:SoundnessRecursivePlay} still holds, i.e., the specification is not violated in this case.
\end{remark}

\subsection{Comments on Completeness}

Intuitively, the synthesis procedure given in \REFalg{def:F} is complete if, whenever there exists a strategy $\hat{f}$ over 
the game graph $\G$ s.t.\ all plays $\hat{\pi}\in\Play{}$ compliant with $\hat{f}$ induce a 
set of local play sequences that are winning for $\Tuple{[\Glall],\I{},\PhiallR}$ 
(if the environment plays an admissible strategy), then there exists a hierarchical strategy 
$F$ s.t.\ its compliant play $\pi$ generated by \eqref{equ:F} induces projected plays that are also winning 
for $\Tuple{[\Glall],\I{},\PhiallR}$ (if the environment plays an admissible strategy).

Unfortunately, this statement is not true. The major problem arises from the fact that assume-admissibly winning strategies are usually not unique for a particular game. Therefore, using one particular strategy calculated by $\Sol{}{\cdot}$ disregards other winning plays. This has two important consequences. First, a move of the current layer $l$ strategy cannot be revised if the current layer $l-1$ game is not realizable for the corresponding reachability specification, even if there exists a different possibly winning extension in layer $l$. In our robot example, this corresponds to the case where the robot is in a particular room $r$ with two adjacent rooms $r'$ and $r''$, where visiting either of them is winning. Now the current strategy for the room layer deterministically picks room $r'$. If the way towards room $r'$ is blocked by a static obstacle, the game in layer $0$ and context $r$ does not have a solution and the play gets stuck.

This problem also arises in reverse layer interaction, as assume-admissibly winning strategies are only ensured to be winning against 
a \enquote{local} admissible environment strategy. 
They do not consider admissible environment moves in higher layers that might cause specification changes in 
the current layer. 
Hence, the local strategy synthesis might pick a strategy that leads the play 
to a region of the state space which is losing for a different specification that might occur later in this game 
due to such an admissible environment move in a higher layer. In the above example this would correspond to the case that the door to room $r'$ gets closed which is visible to layer $1$ and therefore causes the strategy to request the robot to move to room $r''$, instead. Now assume that the way towards both $r'$ and $r''$ was unblocked initially. Given the specification to reach $r'$ the robot might pick one of two passages which allow to reach $r'$ but the selected one is to narrow for the robot to turn. When the specification changes, the robot cannot turn and approach $r''$, hence the game in layer $0$ and context $r$ does not have a solution and the play gets stuck.
Taking these interactions into account when synthesizing local assume-admissible winning strategies is a 
promising idea for future work to obtain a complete algorithm. 
This would also reduce blocking situations which are caused by this interplay.



Completeness holds in the special case of a trivial environment (which has no choice of moves) and the strategy 
only picks one among the choice of system moves \citep[as e.g.\ in][]{KloetzerBelta_2008,VasileBelta_2014}. 
However, in this case, one can compute a strategy statically using a dynamic programming procedure similar to
context free reachability \citep[see][]{RHS95,AlurLaTorreMadhusudan_2003b}.


%
%
%
%
%
%
%

%

\section{Conclusion}\label{sec:Conclusion}
We have shown in this paper how a large-scale reactive controller synthesis problem with intrinsic \emph{hierarchy} and \emph{locality} can be modeled as a hierarchical two player game over a set of local game graphs  w.r.t.\ to a set of local strategies on multiple, interacting abstraction layers. 
We have proposed a reactive controller synthesis algorithm for such hierarchical games that allows for \emph{dynamic specification changes} 
at each step of the play which is recalculated online in every step. This re-calculation becomes computationally tractable by the proposed decomposition. 
We have shown that our algorithm is sound: whenever the environment meets its assumptions and all dynamically generated local games have a 
solution, the controller synthesis algorithm generates a winning hierarchical play for a given specification. If these assumptions do not hold, the algorithm terminates but the generated finite play does not violate the specification up to this point. 

\appendix

\section{Additional Lemmas}\label{sec:app:proofs}

%

\begin{lemma}\label{lem:kappallp1}
 Let $\play{}$ be a play and $\kappa^l$ its timescale transformation for level $l\in[0,\lmax]$.
 For all $l\in[0,\lmax-1]$, we have 
 \begin{align*}
 &\AllQ{k\in\dom{\kappa^{l+1}}}{
  \ExQ{m\in\dom{\kappa^{l}}}{\kappa^{l+1}(k)=\kappa^{l}(m)}}
  \quad\text{and}\quad\maxw{\kappa^{l}}\geq\maxw{\kappa^{l+1}}.
 \end{align*}
\end{lemma}

\begin{proof}
We prove both statements by contradiction. \\
%
%
Take $k\in\dom{\kappa^{l+1}}$ and define $n=\kappa^{l+1}(k)$. 
Assume that there exists no $m\in\dom{\kappa^{l}}$ s.t.\ 
$n=\kappa^{l}(m)$. 
This implies, by the definition of $\kappa^l$ in \eqref{equ:kappa}, that $y^l(n-1)=y^l(n)$. However, 
this implies (by definition of layers) 
that $y^{l+1}(n-1)=y^{l+1}(n)$, which is a contradiction as the 
assumption $n=\kappa^{l+1}(k)$ implies (from \eqref{equ:kappa}) that $y^{l+1}(n-1)\neq y^{l+1}(n)$.\\
Assume that there exists a $k\in\dom{\kappa^{l+1}}$ s.t.\ 
$k>\maxw{\kappa^{l}}$ and $n=\kappa^{l+1}(k)$. 
As before, this implies $y^l(n-1)=y^l(n)$ and hence $y^{l+1}(n-1)=y^{l+1}(n)$ which 
is a contradiction to the assumption that $k\in\dom{\kappa^{l+1}}$.
\end{proof}

%
%
%
%
\begin{lemma}\label{lem:playGl}
For each game $\G$, each play $\play{}$ of $\G$ and each $l\in[0,\lmax]$, we have
 \begin{align}\label{equ:lem:playGl}
  \AllQ{m\in\dom{\plP{l}},n\in(\kappa^l(m),\kappa^l(m+1)]}{
  \begin{propConjA}
   x^l(n)\in\TrE{l}{}\Tuple{\plPx{l}(m),\plPy{l}(m)}\\
  \plPy{l}(m+1)\in\TrS{l}{}\Tuple{\plPx{l}(m+1),\plPy{l}(m)}
  \end{propConjA}
  }.
 \end{align}
\end{lemma}

\begin{proof}
  Pick $l\in[1,\lmax]$ and $m\in\dom{\plP{l}}$ s.t. $m<\maxk{\kappa^l}$ and $\pi'=\pi\ll{\kappa^l(m),\maxk{\pi}}$ and $\pi''=\pi\ll{\kappa^l(m+1)-1,\maxk{\pi}}$.  Observe that $\pi',\pi''\in\Play{}$ by definition and we denote by $\kappa'^l$ and $\kappa''^l$ their respective timescale transformations defined via \eqref{equ:kappa}. Observe that $m<\maxk{\kappa^l}$ implies $\maxk{\kappa'^l},\maxk{\kappa''^l}>0$. We therefore obviously have $n\in(0,\kappa'^l(1)]$ and observe from the construction of $\pi'$ that
  \begin{align*}
   \pi'^l(\kappa'^l(0))=\pi^l(\kappa^l(m))&=\Tuple{\plPx{l}(m),\plPy{l}(m)}\quad\text{and}\\
   \pi'^l(\kappa'^l(1))=\pi^l(\kappa^l(m+1))&=\Tuple{\plPx{l}(m+1),\plPy{l}(m+1)}.
  \end{align*}
 With this it immediately follows from \eqref{equ:Gl:TrE} that $x^l(n)\in\TrE{l}{}\Tuple{\plPx{l}(m),\plPy{l}(m)}$.
 Observe that $m<\maxk{\kappa^l}$ implies that $\plPy{l}(m)\neq\plPy{l}(m+1)$. It furthermore follows from \eqref{equ:kappa} that $y^l(\kappa^l(m+1)-1)=\plPy{l}(m)$. Using these observations we have
 \begin{align*}
   \pi''^l(\kappa''^l(1)-1)=\play{l}(\kappa^l(m+1)-1)=&\Tuple{x^l(\kappa''^l(1)-1),\plPy{l}(m)}\\
  \text{and}\quad\pi''^l(\kappa''^l(1))=\play{l}(\kappa^l(m+1))=&\Tuple{\plPx{l}(m+1),\plPy{l}(m+1)}.
 \end{align*}
 With this it immediately follows from \eqref{equ:Gl:TrS} that $\plPy{l}(m+1)\in\TrS{l}{}\BR{\plPx{l}(m+1),\plPy{l}(m)}$.
\end{proof}

\begin{lemma}\label{lem:playGly}
Let $[\Glall]$ be a set of LGGs and $\Play{l}_{y}$ the set of plays in $\G^l_{y}$. Furthermore, let $\play{}\in\Play{}$ and $\rhoall$ its induced set of local projected play sequences. Then
it holds
for all $l\in[0,\lmax-1]$ and $m\in\domp{\plP{l+1}}$
that 
 \begin{subequations}
\begin{align}
& 
 \begin{propConjA}
    \AllQ{k\in[\kappa^{l+1}_l(m\mips1),\kappa^{l+1}_l(m))}{\plPy{l}(k)\in\Yli{l}{\plPy{l+1}(m-1)}}\\
    \plPy{l}(\kappa^{l+1}_l(m))\in\INTERSECT{\Yla{l}{\plPy{l+1}(m-1)}}{\Yli{l}{\plPy{l+1}(m)}}\\
  \plPdm{l}(m-1)\in\Play{l}_{\plPy{l+1}(m-1)}
 \end{propConjA}\label{equ:playGly}\\
 \intertext{and for all $l\in[0,\lmax-1]$ that}
  & 
 \begin{propConjA}
    \AllQ{k\in[\kappa^{l+1}_l(\maxk{\plP{l+1}}),\maxk{\plP{l}}]}{\plPy{l}(k)\in\Yli{l}{\maxw{\plPy{l+1}}}}\\
  \maxw{\plPdm{l}}\in\Play{l}_{\maxw{\plPy{l+1}}}
 \end{propConjA}.\label{equ:playGly:end}
\end{align}
 \end{subequations}
\end{lemma}

\begin{proof}
As the proof of \eqref{equ:playGly:end} is a simplified version of the proof for \eqref{equ:playGly}, we only give the latter.
 We fix $l\in[0,\lmax-1]$, $m\in\dom{\kappa^{l+1}}$ and $k\in[\kappa^{l+1}_l(m\mips1),\kappa^{l+1}_l(m))$ and prove all lines of the statement separately. To simplify notation we use $\nu:=\plPy{l+1}(m-1)$ and $\nu':=\plPy{l+1}(m)$.\\
 \begin{inparaitem}[$\blacktriangleright$]
  \item Pick $r:=\kappa^{l}(k)$ and $r':=\kappa^{l+1}(m)$ and observe that $r\in[\kappa^{l+1}(m-1),\kappa^{l+1}(m))$.
  With this choice, \eqref{equ:kappa}, \eqref{equ:projpi} and \eqref{equ:layers} imply
    \begin{subequations}
  \begin{align}
   &y^{l+1}(r)=\nu\neq\nu'=y^{l+1}(r')\label{equ:proof:playGlySound:00},\\
 &y^{l+1}(r)=\Ry{l+1}(y^{l}(r))~\text{and}~y^{l+1}(r')=\Ry{l+1}(y^{l}(r')).\label{equ:proof:playGlySound:0}
\end{align}
\end{subequations}
Substituting $y^{l}(r)=\plPy{l}(k)$ and $y^{l}(r')=\plPy{l}(\kappa^{l+1}_l(m))$ in \eqref{equ:proof:playGlySound:0} and using \eqref{equ:Gly:Yli} gives
   \begin{align}\label{equ:proof:playGlySound:1}
    \plPy{l}(k)\in\Yli{l}{\nu},\quad\text{and}\quad
   \plPy{l}(\kappa^{l+1}_l(m))\in\Yli{l}{\nu'},
  \end{align}
  where the left side of \eqref{equ:proof:playGlySound:1} proves the first line of \eqref{equ:playGly}.\\
  \item Recall from \REFprop{prop:piInPlay} that $\plP{l}\in\Play{l}$.
  Using \REFdef{def:Gl} this implies that
   \begin{subequations}
  \begin{align}\label{equ:proof:playGlySound:2}
   \plPx{l}(k+1)\in\TrE{l}{}\BR{\plPx{l}(k),\plPy{l}(k)}\quad\text{and}\quad
   \plPy{l}(k+1)\in\TrS{l}{}\BR{\plPx{l}(k+1),\plPy{l}(k)}.
  \end{align}
  Using the left side of \eqref{equ:proof:playGlySound:1} and \REFass{ass:Xbinlnu}, \eqref{equ:proof:playGlySound:2} implies
  \begin{align}\label{equ:proof:playGlySound:2c}
   \plPy{l}(k+1)\in\TrS{l}{}\BR{\rx{l}{\nu}(\plPx{l}(k+1)),\plPy{l}(k)}=\TrS{l}{}\BR{\plPxd{l}(k+1),\plPy{l}(k)}.
  \end{align}
   As $\plPxd{l}(k+1)=\rx{l}{\nu}(\plPx{l}(k+1))\in\X{l}_{\nu}$ (from \eqref{equ:Gly:X}) it follows from \eqref{equ:proof:playGlySound:2c} and \eqref{equ:Yla} that 
  \begin{equation}\label{equ:proof:playGlySound:2d}
   \plPy{l}(\kappa^{l+1}_l(m))\in\Yla{l}{\nu}.
  \end{equation}
   \end{subequations}
  Combining \eqref{equ:proof:playGlySound:2d} with the right side of \eqref{equ:proof:playGlySound:1} proves the second line of \eqref{equ:playGly}.\\
  \item Using \eqref{equ:proof:playGlySound:1}, \eqref{equ:proof:playGlySound:2d}, \eqref{equ:Gly:X} and \eqref{equ:proof:playGlySound:2} in \eqref{equ:Gly:Tr} implies that 
  \begin{align}\label{equ:proof:playGlySound:3}
   \plPxd{l}(k+1)\in\TrE{l}{\nu}\BR{\plPxd{l}(k),\plPy{l}(k)}~\text{and}~
   \plPy{l}(k+1)\in\TrS{l}{\nu}\BR{\plPxd{l}(k+1),\plPy{l}(k)},
  \end{align}
  hence, the third line of \eqref{equ:playGly} holds.
 \end{inparaitem}
\end{proof}

\begin{lemma}\label{lem:GotStuck:1}
Let $\pi$ be a \emph{maximal} play computed by \eqref{equ:F}. Then it holds for all $l\in[0,L-1]$ and $k\in\dom{\pi}$ that 
 \begin{equation}\label{proof:thm:GotStuck:1}
  \propAequ{
  \begin{propConjA}\neg\UnReal{l}(k)\\\neg\GotStuck{l+1}(k)\end{propConjA}}
  {(\plg{l}(k),\plxd{l}(k+1))\in\dom{f^{l}(k)}}
 \end{equation}
 if $\neg\Done{l+1}(k)$.
\end{lemma}
\begin{proof}
 \begin{inparaitem}[$*$]
  \item[\enquote{$\Rightarrow$}] The left side of \eqref{proof:thm:GotStuck:1} and \eqref{equ:F:UnReal} implies $\neg\GotStuck{l}(k)$ and $\neg\Done{l+1}(k)$ implies $\neg\Done{l}(k)$ from \eqref{equ:F:Done}. Using both observations in \eqref{equ:F:GotStuck} implies $(\plg{l}(k),\plxd{l}(k+1))\in\dom{f^{l}(k)}$.\\
\item[\enquote{$\Leftarrow$}] The right side of \eqref{proof:thm:GotStuck:1} implies $f^l(k)\neq\emptyset$. Therefore, it follows from \eqref{equ:F:flk} that $\neg\GotStuck{l+1}(k)$ and (as $\neg\Done{l}(k)$) from \eqref{equ:F:GotStuck} $\neg\GotStuck{l}(k)$. Using both observations in \eqref{equ:F:UnReal} also gives $\neg\UnReal{l}(k)$.
\end{inparaitem}
\end{proof}

\begin{lemma}\label{lem:GotStuck:UnRealCanNotWait}
  Let $\pi$ be a \emph{maximal} play computed by \eqref{equ:F}. Then it holds for all $k\in\dom{\pi}$ that
\begin{equation}\label{equ:GotStuck:UnRealCanNotWait}
 \propAequ{
\ExQ*{l\in[0,L]}{\UnReal{l}(k)}\\
 }
 {\GotStuck{0}(k)}.
\end{equation}
\end{lemma}
\begin{proof}
 \begin{inparaitem}[$\bullet$]
\item[\enquote{$\Rightarrow$}:]
Pick $l$ s.t. $\UnReal{l}(k)$ and observe that this implies $\GotStuck{l}(k)$ (from \eqref{equ:F:UnReal}) and hence $\neg\Done{l}(k)$ (from \eqref{equ:F:GotStuck}).
Using the first line of \eqref{equ:F:flk} this implies $f^{l-1}(k)=\emptyset$. As $\neg\Done{l}(k)$ also implies $\neg\Done{l-1}(k)$ from \eqref{equ:F:Done} it follows from \eqref{equ:F:GotStuck} that $\GotStuck{l-1}(k)$ is true (i.e., $\propImp{\GotStuck{l}(k)}{\GotStuck{l-1}(k)}$). Applying this reasoning repetitively we eventually obtain $\GotStuck{0}(k)$.\\
\item[\enquote{$\Leftarrow$}:] 
Using  \eqref{equ:F:UnReal}, $\GotStuck{0}(k)$ implies that the right side of \eqref{proof:thm:GotStuck:1} in \REFlem{lem:GotStuck:1} is false. Hence, either $\UnReal{0}(k)$ or $\GotStuck{1}(k)$ is true. If \UnReal{0} is true the statement is proven. We therefore assume that $\GotStuck{1}(k)$ is true. We can reuse the same reasoning to either eventually get \UnReal{l} for some $l\in[0,L]$ (what proves the statement) or reach $\GotStuck{L}(k)$. However, it follows from \eqref{equ:F:UnReal} that the latter is equivalent to \UnReal{L}, what proves the statement.
\end{inparaitem}
\end{proof}


\begin{lemma}\label{lem:CondReach}
  Let $\pi=\Tuple{x,y}\in\Play{l}_\nu$ for some $\nu\in\Y{l+1}$ s.t. $y(0)\in\Yli{l}{\nu}$ and $\psi^l_\nu(\nu')$ with $\nu'\in\Y{l+1}, ~\nu\neq\nu'$ as in \eqref{equ:psi}.
  Then it holds that 
  \begin{align}\label{equ:CondReach}
   \propAequ{
    \pi\in\psi^l_\nu(\nu')
   }{
   \begin{propConjA}
    \maxk{\pi}<\infty\\
    \AllQ{k<\maxk{\pi}}{y(k)\in\Yli{l}{\nu}}\\
    \maxw{y}\in\INTERSECT{\Yla{l}{\nu}}{\Yli{l}{\nu'}}
   \end{propConjA}}.
  \end{align}
\end{lemma}

\begin{proof}
 \begin{inparaitem}
  \item[\enquote{$\Leftarrow$}]
  \eqrefR{equ:CondReach}{1} and \eqrefR{equ:CondReach}{3} immediately imply that $\pi\in\psi^l_\nu(\nu')$ (from the first line of \eqref{equ:psi}).
  \item[\enquote{$\Rightarrow$}]
  \eqrefR{equ:CondReach}{2} is the only non-obvious conclusion from \eqref{equ:psi}.
  Recall that $\pi\in\Play{l}_\nu$ and $y(0)\in\Yli{l}{\nu}$. Therefore it holds for all $r\leq\maxk{\pi}$ that $y^{l}_{\nu}(r)\in\UNION{\Yli{l}{\nu}}{\Yla{l}{\nu}}$. Now assume that there exists $r'<\maxk{\pi}$ s.t. $y(r')\in\Yla{l}{\nu}$.
  Using \eqref{equ:Gly:Y} this would imply $y(r')\notin\Yli{l}{\nu}$ and therefore from \eqref{equ:Gly:TrS} there exist no $\tilde{x},\tilde{y}$ s.t. $\tilde{y}\in\TrS{l}{\nu}\BR{\tilde{x},y(r')}$, implying $r'=\maxk{\pi}$ which is a contradiction to the assumption.
 \end{inparaitem}
\end{proof}

\begin{lemma}\label{lem:DoneOnlyEnding}
   Let $\pi$ be a  play computed by \eqref{equ:F} up to length $\maxk{\pi}$. Then it holds for all $l\in[1,L]$ and $k<\kappa^{l}(\maxk{\plP{l}})$ that $\neg\Done{l}$. 
\end{lemma}
\begin{proof}
 We prove the statement by contradiction. Pick any $l\in[1,L]$ and $k<\kappa^{l}(\maxk{\plP{l}})$ and assume that $\Done{l}$ is true.
 First observe that this implies $\Done{l'}$ for all $l'\in[l,L]$. With this it follows from \eqref{equ:F:flk} that $f^l(k)=h^l$. Now using \eqref{equ:F:Done} this implies that $(\plg{l}(k),\plxd{l}(k+1))\notin \dom{f^l(k)}$ and therefore the play would not be able to leave the current context. This is a contradiction to the assumption that $k<\kappa^{l}(\maxk{\plP{l}})$, what proves the statement. 
\end{proof}

\begin{lemma}\label{lem:SoundnessRecursivePlay:A}
  Let $\pi$ be a  play computed by \eqref{equ:F} up to length $\maxk{\pi}$. Then it holds for all $l\in[0,\lmax]$ and $m\in\domp{\plP{l}}$ that 
\begin{subequations}\label{equ:proof:linit}
  \begin{align}
 \plPxd{l}(m)&\in\TrE{l}{\plPy{l}(m-1)}(\plPxd{l}(m-1),\plPy{l}(m-1))~\text{and}\label{equ:proof:linit:a}\\
 \plPy{l}(m)&=f^{l}(\kappa^{l}(m)-1)(\plg{l}(\kappa^{l}(m)-1),\plPxd{l}(m)).\label{equ:proof:linit:b}
\end{align}
\end{subequations}

\end{lemma}

\begin{proof}
\begin{subequations}
Recall that $\play{}\in\Play{}$ from \REFprop{prop:piInPlay}. Therefore, \eqref{equ:proof:linit:a} follows directly from \eqref{equ:proof:playGlySound:3} in \REFlem{lem:playGly} (see \REFapp{sec:app:proofs}). We show \eqref{equ:proof:linit:b} by induction.\\
\begin{inparaitem}[$\blacktriangleright$]
 \item $l=0$:\\
 Recall that \eqref{equ:proof:l0} holds for $l=0$. As $\kappa^0$ is the identity map, the second line in \eqref{equ:proof:l0} and \eqref{equ:proof:linit} is equivalent for $l=0$.\\ 
 \item $l\fun l+1$:\\
 \begin{inparaitem}[$\bullet$]
  \item Pick $m\in\domp{\plP{l+1}}$, $k=\kappa^{l+1}(m)$, $\nu:=\plPy{l+1}(m-1)$ and $\nu':=\plPy{l+1}(m)$ 
  and recall from \REFlem{lem:kappallp1} that there exists $r\in\Nb$ s.t. $r=\kappa^{l+1}_l(m)$ and $\kappa^l(r)=k$, implying (from \eqref{equ:kappa}) that
 \begin{align}
  &y^{l+1}(k\mips1)=\nu\neq\nu'=y^{l+1}(k),~y^{l}(k)=\plPy{l}(r),\label{equ:proof:A}\\
  &\plPxd{l+1}(m)=\plxd{l+1}(k),~\text{and}~\plPxd{l}(r)=\plxd{l}(k).\notag
 \end{align}\\
 Now it follows from \REFlem{lem:playGly} that $\plPdm{l}(m-1)\in\Play{l}_{\nu}$, hence \REFlem{lem:CondReach} holds for $\plPdm{l}(m-1)$. Now using the first and second line of \eqref{equ:playGly} in \REFlem{lem:CondReach} immediately implies 
 \begin{equation}\label{equ:proof:0A}
  \plPdm{l}(m-1)\in\psi^{l}_{\nu}(\nu').
 \end{equation}
 \item We now show that $\plPdm{l}(m-1)$ is compliant with $f^l(k-1)$:\\
 As \eqref{equ:proof:linit} holds for $l$ 
we know that for all $k'<\maxk{\pi}$ we have  $\neg\GotStuck{l+1}(k')$ (from \REFlem{lem:GotStuck:1}). As additionally $\neg\Done{l+1}(k')$ from \REFlem{lem:DoneOnlyEnding}, \eqref{equ:F:flk} gives that 
 \begin{align}
  &f^{l}(k')=f^{l}_{y^{l+1}(k')\nu^{l+1}}(k')~\label{equ:proof:Ba}\\
   &\text{s.t.}~ \nu^{l+1}(k')=f^{l+1}(k')(\plg{l+1}(k'), \plxd{l+1}(k'+1)).\label{equ:proof:B}
 \end{align}
Now pick $s$ s.t. 
 \begin{equation}\label{equ:proof:C3}
  \AllQ{k',k''\in[\kappa^{l}(r-s),\kappa^{l}(r)-1]}{y^{l+1}(k')= y^{l+1}(k'')\wedge\nu^{l+1}(k')=\nu^{l+1}(k'')},
 \end{equation}
with $\nu^{l+1}$ as in \eqref{equ:proof:B} and observe that this implies $\kappa^{l}(r-s)\in[\kappa^{l+1}(m-1),\kappa^{l+1}(m))$. Using \eqref{equ:proof:C3} in \eqref{equ:F:flk:DC} therefore gives for all $k'\in[\kappa^{l}(r-s),\kappa^{l}(r)-1]$ that
\begin{equation}\label{equ:proof:C4}
 f^{l}(k')=f^{l}(k-1)=\Sol{}{\G^l_{\nu},\Set{\plg{l}(\kappa^{l}(r-s))},\phisconc{l}{\nu}{\nu'^{l+1}(k-1)}}.
\end{equation}
As \eqref{equ:proof:linit} holds for $l$ we can therefore substitute $f^{l}(\kappa^{l}(r)-1)$ in \eqref{equ:proof:linit:b} by $f^{l}(k-1)$ and obtain for all $r'\in[r-s,r-1]$ that 
\begin{equation}\label{equ:proof:C5}
 \plPy{l}(r')=f^{l}(k-1)(\plg{l}(\kappa^{l}(r')-1),\plPxd{l}(r')).
\end{equation}
It furthermore follows from the construction of $\plg{l}$ in \eqref{equ:F:LH} and $\plPdm{l}$ in \eqref{def:plPdm} that
 \begin{align}\label{equ:proof:C2}
  \plPdm{l}(m-1)&=\plg{l}(\kappa^{l}(r-s))\sconc \plPd{l}\ll{r-s+1,r}
 \end{align}
Now pick $n=\maxk{\plg{l}(\kappa^{l}(r-s))}$ and observe that $\plPdm{l}(m-1)\ll{0,n}\in\Set{\plg{l}(\kappa^{l}(r-s))}$. Additionally using \eqref{equ:proof:C5} therefore implies that $\plPdm{l}(m-1)\in\Compl\Tuple{f^l(k-1),\Set{\plg{l}(\kappa^{l}(r-s))}}$ (from \eqref{equ:newcompliant}). Using \eqref{equ:proof:C4} and \eqref{equ:WeaklyWinningNonAdm} from \REFprop{prop:WeaklyWinning} it follows that
\begin{equation}\label{equ:proof:D}
\plPdm{l}(m-1)\in\overline{\phisconc{l}{\nu}{\nu^{l+1}(k-1)}}.
\end{equation}
 %
\item It remains to shown that $\nu^{l+1}(k-1)=\nu'(=\plPy{l+1}(m)=y^{l+1}(k))$:\\
Using the fact that $y^l(k)\in\Yla{l}{\nu}$ it follows from \REFlem{lem:CondReach} and \eqref{equ:ConcatSpec} that \eqref{equ:proof:0A} and \eqref{equ:proof:D} can only be satisfied simultaneously if
\begin{equation}\label{equ:proof:E}
\plPdm{l}(m-1)\in\phisconc{l}{\nu}{\nu^{l+1}(k-1)}\quad\text{and}\quad\nu^{l+1}(k-1)=\nu'.
\end{equation}
With this observation \eqref{equ:proof:linit:b} immediately follows for $l+1$ from \eqref{equ:proof:B} as $\nu^{l+1}(k-1)=\plPy{l+1}(m)$.
%
\end{inparaitem}
\end{inparaitem}
\end{subequations}
\end{proof}

\begin{lemma}\label{lem:localWin}
 Let $\pi$ be a  play computed by \eqref{equ:F} up to length $\maxk{\pi}$ and $\rhoall$ its induced set of local projected play sequences. Then
it holds for all $l\in[0,\lmax-1]$ and $m\in\domp{\plP{l+1}}$ that 
\begin{subequations}\label{equ:localWin}
  \begin{align}
  \plPdm{l}(m-1)&\in\phisconc{l}{\plPy{l+1}(m-1)}{\plPy{l+1}(m)}\label{equ:localWin:a}
  \end{align}
and for $m=\maxk{\plP{l+1}}$ there exists $\nu'\in\ON{Post}^{l+1}(\plPy{l+1}(m))$ s.t.
\begin{align}
   \plPdm{l}(m)&\in\overline{\phisconc{l}{\plPy{l+1}(m)}{\nu'}}.\label{equ:localWin:b}
\end{align}
\end{subequations}
\end{lemma}

\begin{proof}
 \eqref{equ:localWin:a} follows from \eqref{equ:proof:E} in the proof of \REFlem{lem:SoundnessRecursivePlay:A}. We prove \eqref{equ:localWin:b}:\\
 Pick $l\in[0,\lmax-1]$ and $m=\maxk{\plP{l+1}}$ and recall from \REFlem{lem:SoundnessRecursivePlay:A} that \eqref{equ:proof:linit} holds for all $k\in\domp{\plP{l}}$. Therefore $\neg\GotStuck{l+1}(k')$ (from \REFlem{lem:GotStuck:1}) for all $k'<\maxk{\pi}$. Now we have two cases.\\
 \begin{inparaenum}[(i)]
  \item If $\neg\Done{l+1}(\kappa^{l+1}(m))$, \eqref{equ:proof:B} in the proof of \REFlem{lem:SoundnessRecursivePlay:A} holds for $k'\in[\kappa^{l+1}(m),\allowbreak\kappa^{l}(\maxk{\plP{l}})]$. Following exactly the same reasoning as in \eqref{equ:proof:B}-\eqref{equ:proof:D} we obtain 
 \begin{equation*}
  \plPdm{l}(m)\in\overline{\phisconc{l}{\plPy{l+1}(m)}{\nu^{l+1}(k-1)}}
 \end{equation*}
 with $\nu^{l+1}(k-1)$ as in \eqref{equ:proof:B}, implying \eqref{equ:localWin:b}.\\
 \item If $\Done{l+1}(\kappa^{l+1}(m))$, it follows from \eqref{equ:F:flk} and \eqref{equ:F:hlk:DC} that for $k'\in[\kappa^{l+1}(m),\allowbreak\kappa^{l}(\maxk{\plP{l}})]$
 \begin{equation}\label{equ:proof:hl}
  f^l(k')=h^l(\kappa^{l+1}(m))=\Sol{}{\G^l_{\plPy{l+1}(m)},\Set{\plg{l}(\kappa^{l+1}(m))},\varphi^l_{\plPy{l+1}(m)},\eass^l_{\nu}}
 \end{equation}
 and from the construction of $\plg{l}$ and $\plPdm{l}$ in \eqref{equ:F:LH} and \eqref{def:plPdm} that
  $\plg{l}(\kappa^{l+1}(m))=\maxw{\plPdm{l}(m-1)}$. 
By substituting \eqref{equ:proof:hl} in \eqref{equ:proof:linit:b} we therefore obtain $\plPdm{l}(m)\in\Compl(h^l(\kappa^{l+1}(m)),\maxw{\plPdm{l}(m-1)})$ 
(from \eqref{equ:newcompliant}). Using \eqref{equ:proof:C4} and \eqref{equ:WeaklyWinningNonAdm} from \REFprop{prop:WeaklyWinning} it follows that
$\plPdm{l}(m)\in\overline{\varphi^l_{\plPy{l+1}(m)}}$. Now recall from \eqref{equ:ConcatSpec} that $\overline{\varphi^l_{\plPy{l+1}(m)}}\subseteq\overline{\phisconc{l}{\plPy{l+1}(m)}{\nu^{l+1}(k-1)}}$, what proves the statement.
 \end{inparaenum}
\end{proof}


\begin{lemma}\label{lem:GotStuck:Exist:A}
Let $\pi$ be a \emph{maximal and environment admissible} play computed by \eqref{equ:F} s.t. \eqref{equ:ass2} holds. Then it holds that
\begin{align*}
\propImp{
\ExQ*{k\in\dom{\pi},l\in[0,L]}{\Done{l}(k)}
}{\ExQ*{k'\in\dom{\pi},k'\geq k}{\Done{0}(k')}}.
\end{align*}
\end{lemma}
\begin{proof}
Pick $k\in\dom{\pi},l\in[0,L]$ s.t. $\Done{l}(k)$ and assume $l>0$ as for $l=0$ the statement follows trivially.
Giving $\Done{l}(k)$, \eqref{equ:F:GotStuck} implies $\neg\GotStuck{l}(k)$ and \eqref{equ:F:flk} implies $f^{l-1}(k)=h^{l-1}(k)$.
Giving $\neg\UnReal{l-1}(k)$ and $\neg\GotStuck{l}(k)$, \eqref{equ:F:UnReal} implies $\neg\GotStuck{l-1}(k)$ and therefore (from \eqref{equ:F:GotStuck}) either $\Done{l-1}(k)$ or there exists a next step according to $h^{l-1}(k)$. Assume the latter is true.
Recall from \eqref{equ:F:LH} that $h^{l-1}(k)$ is an assume admissible winning strategy for the game $\Tuple{\G^l_{y^{l}(k)},\Set{\plg{l-1}(\maxw{\kappa^{l-1}})},\varphi^l_{y^{l}(k)}}$ and from \eqref{equ:PhiallR} that $\varphi^l_{y^{l}(k)}$ only contains finite strings. If the environment plays admissible, we therefore eventually obtain $\Done{l-1}(k')$ with $k<k'<\infty$.
Applying this reasoning iteratively, eventually leads to $\Done{0}(k'')$ where the time between $k$ and $k''$ is ensured to be finite.
\end{proof}

\begin{lemma}\label{lem:GotStuck:Exist:B}
Let $\pi$ be a \emph{maximal and environment admissible} play computed by \eqref{equ:F} s.t. \eqref{equ:ass2} holds. 
Then it holds that 
\begin{align*}
&\AllQ*{k\in\dom{\pi},l\in[0,L]}{\neg\Done{l}(k)}
\Leftrightarrow\AllQ*{l\in[0,L]}{|\plP{l}|=\infty}
\Leftrightarrow\BR{|\pi|=\infty}.
\end{align*}
\end{lemma}

\begin{proof}
We show this proof in two steps.\\
\begin{inparaitem}[$\blacktriangleright$]
\item Show $\propAequ{\AllQ*{k\in\dom{\pi},l\in[0,L]}{\neg\Done{l}(k)}}{\BR{|\pi|=\infty}}$:\\
Using \eqref{equ:ass2} in \eqref{equ:thm:GotStuck} of \REFthm{thm:GotStuck} gives
\begin{align}\label{equ:proof:lem:GotStuck:Exist:1}	
 \propAequ{\ExQ{l\in[0,L]}{\AllQ{k\in\dom{\pi}}{\neg\Done{l}(k)}}}{|\pi|=\infty},
\end{align}
immediately implying the \enquote{$\Rightarrow$} part of the statement. Now we prove the \enquote{$\Leftarrow$} part by contradiction. 
Assume that there exists $l\in[0,L], k\in\dom{\pi}$ s.t. $\Done{l}(k)$. Then \REFlem{lem:GotStuck:Exist:A} implies $\Done{0}(k')$. Using (from \eqref{equ:F:Done}) this implies $\Done{l}(k')$ for all $l\in[0,L]$, which gives a contradiction as the left side of \eqref{equ:proof:lem:GotStuck:Exist:1} holds from $\BR{|\pi|=\infty}$.\\
%
%
 \item Show $\propAequ{\AllQ{l\in[0,L]}{|\plP{l}|=\infty}}{|\pi|=\infty}$:\\
 First observe that \enquote{$\Rightarrow$} trivially holds as $\plP{0}=\pi$. We prove \enquote{$\Leftarrow$} by contradiction.
 Assume there exists $l\in[0,L]$ s.t. $|\plP{l}|<\infty$, i.e., with $k=\maxk{\plP{l}}$ we have $(\plg{l}(k),\plxd{l}(k+1))\notin\dom{f^{l}(k)}$. Now recall from the first part of this proof that $|\pi|=\infty$ implies $\neg\Done{l'}(k)$ for all $l'\in[0,L]$ and \eqref{equ:ass2} implies $\neg\UnReal{l+1}(k)$.  
 Then it follows from \REFlem{lem:GotStuck:1} that $\GotStuck{l+1}(k)$. With this $\GotStuck{l}(k)$ (from \eqref{equ:F:flk}) and therefore eventually $\GotStuck{0}(k)$, which implies $|\pi|<\infty$ with $\maxk{\pi}=k$, which is a contradiction to the assumption.
\end{inparaitem}
\end{proof}
%
%
%
%
%
%



%


\end{document}